\documentclass{eptcs}

\usepackage{bookmark}   
\usepackage[UKenglish]{babel}
\usepackage{colorprofiles}
\hypersetup{
  colorlinks=true,
  citecolor=gray,
  linkcolor=blue,
}

\usepackage{csquotes}

\usepackage{import}
\usepackage{subfiles}

\usepackage{etoolbox}
\apptocmd{\sloppy}{\hbadness 10000\relax}{}{} 

\usepackage{mathtools}
\usepackage{physics}
\usepackage{amsfonts}
\usepackage{mathdots}   
\usepackage{mathrsfs}   
\usepackage{bigints}    
\usepackage{xfrac}      
\usepackage{stmaryrd}   
\SetSymbolFont{stmry}{bold}{U}{stmry}{m}{n}   

\usepackage{amsthm}
\usepackage{thmtools}
\usepackage{thm-restate}

\newtheorem{theorem}{Theorem}
\newtheorem{lemma}{Lemma}
\newtheorem{remark}{Remark}

\newtheorem{proposition}{Proposition}

\usepackage{placeins}

\usepackage{bm}   

\usepackage{blkarray}   

\usepackage{graphicx}
\graphicspath{figures}

\usepackage{tikzfig}
\usepackage{pgfplots} 
\pgfplotsset{compat=newest, compat/show suggested version=false} 
\usetikzlibrary{backgrounds}

\tikzstyle{boxx}=[draw=black, shape=rectangle, fill=white, minimum size=.95em, inner sep=0.15em, scale=0.85, font={\scriptsize}]
\tikzstyle{gbox}=[boxx, draw=black, shape=rectangle, fill={zx_green}, tikzit fill={rgb,255: red,181; green,215; blue,181}, tikzit shape=rectangle]
\tikzstyle{hbox}=[boxx, draw=black, shape=rectangle, fill=yellow, minimum size=.55em]
\tikzstyle{gn}=[draw=black, shape=circle, fill={zx_green}, inner sep=0.7mm, minimum width=0pt, minimum height=0pt, tikzit fill={rgb,255: red,181; green,215; blue,181}]
\tikzstyle{gn_phase}=[shape=rectangle, fill={zx_green}, draw=black, minimum size=1.2em, rounded corners=0.5em, inner sep=0.2em, outer sep=-0.2em, scale=0.8, font={\footnotesize\boldmath}, tikzit shape=circle, tikzit fill={rgb,255: red,130; green,188; blue,130}]
\tikzstyle{rn}=[gn, fill={zx_red}, draw=black, tikzit fill={rgb,255: red,255; green,165; blue,165}]
\tikzstyle{rn_phase}=[{gn_phase}, fill={zx_red}, draw=black, tikzit fill={rgb,255: red,215; green,96; blue,96}]
\tikzstyle{lmat}=[shape=signal, signal to=west, signal from=east, fill={zx_grey}, draw=black, minimum height=6pt, inner sep=.75pt, font={\scriptsize \boldmath}, tikzit fill=gray, tikzit category=GLA]
\tikzstyle{rmat}=[lmat, shape=signal, signal to=east, signal from=west, tikzit fill=gray, tikzit category=GLA]
\tikzstyle{dmat}=[lmat, shape=signal, signal to=west, signal from=east, tikzit fill=gray, tikzit category=GLA, rotate=90]
\tikzstyle{umat}=[lmat, shape=signal, signal to=east, signal from=west, tikzit fill=gray, tikzit category=GLA, rotate=90]
\tikzstyle{d_split}=[shape=trapezium, fill=white, draw=black, inner sep=0pt, trapezium stretches body, text width=15pt, text height=7pt]
\tikzstyle{d_merge}=[{d_split}, shape=trapezium, draw=black, rotate=180]
\tikzstyle{sd_split}=[shape=trapezium, fill=white, draw=black, inner sep=0pt, trapezium stretches body, text width=10pt, text height=5pt]
\tikzstyle{sd_merge}=[{sd_split}, shape=trapezium, draw=black, rotate=180]
\tikzstyle{ZXdimension}=[font={\tiny}, auto, color=gray]
\tikzstyle{rZWdimension}=[fill={gray!10}, font={\tiny}, draw={gray!30}, rounded rectangle, rounded rectangle east arc=none, text={black!60}, inner sep=1.5pt, anchor=east, tikzit fill=gray]
\tikzstyle{lZWdimension}=[fill={gray!10}, font={\tiny}, draw={gray!30}, rounded rectangle, rounded rectangle west arc=none, text={black!60}, inner sep=1.5pt, anchor=west, tikzit fill=gray]
\tikzstyle{wire label}=[font={\tiny}, auto]
\tikzstyle{control}=[draw=black, shape=circle, fill=black, inner sep=0.5mm]
\tikzstyle{uw}=[shape=isosceles triangle, isosceles triangle stretches=true, fill=black, draw=black, minimum width=0.4cm, minimum height=3mm, inner sep=1pt, shape border rotate=90]
\tikzstyle{dw}=[shape=isosceles triangle, isosceles triangle stretches=true, fill=black, draw=black, minimum width=0.4cm, minimum height=3mm, inner sep=1pt, shape border rotate=-90]

\tikzstyle{braceedge}=[-, decorate, decoration={brace, amplitude=2mm, raise=-1mm}]
\tikzstyle{dotsedge}=[-, dotted, decoration={brace, amplitude=2mm, raise=-1mm}]

\tikzstyle{black dot}=[inner sep=0.5mm,minimum width=0pt,minimum height=0pt,fill=black,draw=black,shape=circle]
\tikzstyle{dot}=[black dot]
\tikzstyle{white dot}=[inner sep=0.5mm,minimum width=0pt,minimum height=0pt,fill=white,draw=black,shape=circle]
\tikzstyle{white phase dot}=[minimum size=3.2mm, font={\footnotesize\boldmath}, shape=rectangle, rounded corners=1.3mm, inner sep=0.8mm, outer sep=-2mm, scale=0.8, draw=black, fill=white]
\tikzstyle{ket}=[minimum size=2mm, font={\tiny\boldmath}, shape=rectangle, rounded corners=0.9mm, inner sep=0.6mm, outer sep=-2mm, scale=0.8, draw=black, fill=black, text=white]
\tikzstyle{box}=[rectangle,fill=white,draw=black, font=\scriptsize, inner sep=2pt]
\tikzstyle{box-no-outline}=[rectangle, inner sep=2pt]
\tikzstyle{fswap}=[inner sep=0.7mm,minimum width=0pt,minimum height=0pt,draw=purple,shape=circle]
\tikzstyle{compact dash}=[dash pattern={on 2pt off 1pt}]
\tikzstyle{W-1-n}=[draw, trapezium, trapezium left angle=70, trapezium right angle=70, fill=black, inner sep=1.8pt]

\definecolor{zx_grey}{RGB}{211,211,211}
\definecolor{zx_red}{RGB}{232,165,165}
\definecolor{zx_pink}{RGB}{255, 130, 160}
\definecolor{zx_green}{RGB}{216,248,216}

\usepackage{xtab}
\usepackage{tabu}

\usepackage{float}

\usepackage[capitalise, noabbrev]{cleveref}

\newcommand{\N}{\mathbb{N}}

\ifdef{\C}
  {\renewcommand{\C}{\mathbb{C}}}
  {\newcommand{\C}{\mathbb{C}}}
\newcommand{\minu}{\texttt{-}}
\newcommand{\plus}{\texttt{+}}

\makeatletter
\newcommand{\oset}[3][1.5ex]{%
  \mathrel{\mathop{#3}\limits^{
    \vbox to#1{\kern-2\ex@
    \hbox{$\scriptstyle#2$}\vss}}}}
\makeatother

\newcommand{\interp}[1]{\left\llbracket#1\right\rrbracket}
\newcommand{\lst}[1]{\left(#1\right)}

\newcommand{\ZXf}{\mathbf{ZX_{f}}}
\newcommand{\ZWf}{\mathbf{ZW_{\!f}}}
\newcommand{\FHilb}{\mathbf{FHilb}}
\newcommand{\interpTo}{\overset{\interp{\cdot}}{\longmapsto}}
\newcommand{\iXW}[1][.]{\interp{#1}_{XW\!}}
\newcommand{\iWX}[1][.]{\interp{#1}_{W\!X}}
\newcommand{\ZXtoZW}{\oset{\iXW}{\longmapsto}}
\newcommand{\ZWtoZX}{\oset{\iWX}{\longmapsto}}
\newcommand{\ZXfromXW}{\oset{\iWX}{\longmapsfrom}}

\newcommand{\sNf}{\sqrt{\!N!}}

\newcommand{\bR}{\begin{color}{red}}
\newcommand{\bB}{\begin{color}{blue}}
\newcommand{\bM}{\begin{color}{magenta}}
\newcommand{\bC}{\begin{color}{cyan}}
\newcommand{\bW}{\begin{color}{white}}
\newcommand{\bBl}{\begin{color}{black}}
\newcommand{\bG}{\begin{color}{green}}
\newcommand{\bY}{\begin{color}{yellow}}
\newcommand{\e}{\end{color}}

\newcommand{\tikzrefsize}[1]{\scriptsize{#1}}
\newcommand{\axref}[1]{\tikzrefsize{\eqref{rule:#1}}}
\newcommand{\lemref}[1]{\tikzrefsize{\eqref{#1}}}

\newcommand{\Mod}[1]{\ (\mathrm{mod}\ #1)}

\newcommand{\eq}[2][~~]{#1\underset{\substack{#2}}{=}#1}
\newcommand{\namedLabel}[1]{\tag{\textsc{#1}}\label{rule:#1}\refstepcounter{equation}}

\newcommand{\zxw}[1]{\cite[#1]{poorCompletenessArbitraryFinite2023}}

\renewcommand{\axref}[1]{\tikzrefsize{\eqref{rule:#1}}}
\renewcommand{\lemref}[1]{\tikzrefsize{\eqref{#1}}}

\usepackage{multicol}

\title{ZX-calculus is Complete for Finite-Dimensional Hilbert Spaces}
\def\titlerunning{ZX-calculus is Complete for Finite-Dimensional Hilbert Spaces}
\date{}
\author{
  Boldizsár Poór$\null^{1,2}$ \and
  Razin A.\@ Shaikh$\null^{1,2}$ \and
  Quanlong Wang$\null^{1}$ \and
  \institute{$\null^{1}$Quantinuum, 17 Beaumont Street, Oxford, OX1 2NA, United Kingdom}
  \institute{$\null^{2}$University of Oxford, United Kingdom}
}
\def\authorrunning{B. Poór, R. Shaikh, and Q. Wang}

\begin{document}

\maketitle

\begin{abstract}
  
The ZX-calculus is a graphical language for reasoning about quantum computing and quantum information theory.
As a complete graphical language, it incorporates a set of axioms rich enough to derive any equation of the underlying formalism.
While completeness of the ZX-calculus has been established for qubits and the Clifford fragment of prime-dimensional qudits, universal completeness beyond two-level systems has remained unproven until now.
In this paper, we present a proof establishing the completeness of finite-dimensional ZX-calculus, incorporating only the mixed-dimensional Z-spider and the qudit X-spider as generators.
Our approach builds on the completeness of another graphical language, the finite-dimensional ZW-calculus, with direct translations between these two calculi.
By proving its completeness, we lay a solid foundation for the ZX-calculus as a versatile tool not only for quantum computation but also for various fields within finite-dimensional quantum theory.

\end{abstract}

\section{Introduction}\label{sec:introduction}


The ZX-calculus~\cite{coeckeInteractingQuantumObservables2007, coeckeInteractingQuantumObservables2008} is a graphical language for reasoning about quantum computation and quantum information theory.
Thanks to its model-independent representation of quantum computation and set of intuitive graphical rewrite rules,
the ZX-calculus has found applications across various domains of quantum technologies~\cite{vandeweteringZXcalculusWorkingQuantum2020}.
These domains include quantum error correction~\cite{debeaudrapZXCalculusLanguage2020, kissingerPhasefreeZXDiagrams2022, huangGraphicalCSSCode2023, townsend-teagueFloquetifyingColourCode2023, cowtanCSSCode2024, rodatzFloquetifyingStabiliser2024},
quantum circuit optimization~\cite{duncanGraphtheoreticSimplificationQuantum2020, debeaudrapFastEffective2020, kissingerReducingTcountZXcalculus2020, vandeweteringOptimalCompilationParametrised2024},
measurement-based quantum computing~\cite{duncanRewritingMeasurementBasedQuantum2010, backensThereBackAgain2021, mcelvanneyFlowpreservingZXcalculusRewrite2023},
fusion-based quantum computing~\cite{bombinUnifyingFlavors2024, pankovichFlexibleEntangledState2023, feliceFusionFlow2024},
compilation~\cite{sivarajahTKetRetargetableCompiler2020},
classical simulation~\cite{codsiClassicallySimulatingQuantum2023, camSpeedingQuantumCircuits2023},
and education~\cite{coeckePicturingQuantumProcesses2017, coeckeQuantumPictures2022, dundar-coeckeQuantumPicturalismLearning2023}.


While the ZX-calculus has primarily been studied for reasoning about qubit quantum computation,
various extensions or modifications allow it to address different aspects of quantum computation and theory.
From the perspective of the ZX-calculus, there are two approaches to modifying the language: changing the set of generators and/or giving them an alternative interpretation.
Languages with a different set of generators include the
ZW-calculus~\cite{coeckeGHZWcalculusContains2011, hadzihasanovicAlgebraEntanglementGeometry2017},
ZH-calculus~\cite{backensZHCompleteGraphical2019, laakkonenPicturingCountingReductions2023},
and ZXW-calculus~\cite{shaikhHowSum2023, wangDifferentiatingIntegrating2024}.
Alternatively, a different interpretation may involve extending the language to higher-dimensional systems such as
qutrits~\cite{wangQutritZXcalculusComplete2018, townsend-teagueSimplificationStrategiesQutrit2022, vandeweteringBuildingQutritDiagonal2023},
quopits~\cite{boothCompleteZXcalculi2022, poorQupitStabiliserZXtravaganza2023, comfortAlgebraStabilizerCodes2023, boothGraphicalSymplecticAlgebra2024},
qudits~\cite{royQuditZHCalculusGeneralised2023, poorCompletenessArbitraryFinite2023, cowtanQuditLatticeSurgery2022},
or the finite-dimensional setting~\cite{wangQufiniteZXcalculusUnified2022, wangCompletenessQufiniteZXW2024, devismeMinimalityFiniteDimensionalZWCalculi2024}.


While languages like the ZXW- and ZW-calculus enable complete reasoning for both qudits and finite-dimensional Hilbert spaces,
defining the ZX-calculus with the same generality can open avenues for interesting applications, as different generator sets exhibit distinct strengths and weaknesses.
For instance, the ZW-calculus~\cite{coeckeThreeQubitEntanglement2011} provides valuable insights for studying multi-partite entanglement~\cite{coeckeThreeQubitEntanglement2011, hadzihasanovicAlgebraEntanglementGeometry2017} and linear optical quantum computing~\cite{defeliceQuantumLinear2023, defeliceLightMatterInteractionZXW2023}.
However, it is less effective for understanding circuit-based quantum protocols.
Similarly, while the ZXW-calculus is a powerful analytical tool capable of expressing Hamiltonians~\cite{shaikhHowSum2023, defeliceLightMatterInteractionZXW2023}, performing differentiation and integration of ZX-diagrams~\cite{wangDifferentiatingIntegrating2024}, classical simulation~\cite{kochContractionZX2024}, and interactions in quantum field theory~\cite{shaikhCategoricalSemanticsFeynman2022}, extracting circuits expressed by these methods remains a hard problem~\cite{debeaudrapCircuitExtractionZXDiagrams2022}.
By contrast, with the ZX-calculus we can keep diagrams easily extractable for quantum circuit optimization~\cite{duncanGraphtheoreticSimplificationQuantum2020}.
Therefore, the finite-dimensional ZX-calculus has the potential to yield results directly applicable to quantum circuits.


When designing graphical languages, there are three crucial properties that ensure the language can be effectively used.
These properties are
(1) \emph{soundness}, which ensures that the rules of the calculus are correct,
(2) \emph{universality}, ensuring that the language can express any element of the underlying formalism, and
(3) \emph{completeness}, enabling the derivation of any equality within the underlying formalism.
Of these, proving completeness poses the most difficulty;
nevertheless, it ensures that the language can be effectively used to graphically derive equalities and proofs.


Indeed, there is a rich history of results investigating the completeness and incompleteness of various calculi.
While the qubit ZX-calculus was first formulated in 2007~\cite{coeckeInteractingQuantumObservables2007},
completeness for qubit quantum computing was not proved until 2017~\cite{ngUniversalCompletionZXcalculus2017, jeandelDiagrammaticReasoningClifford2018}.
This unfolded in several stages, progressively increasing the fragment for which completeness had been proved.
The first completeness proof was for the Clifford fragment~\cite{backensZXcalculusCompleteStabilizer2014}, followed by a proof of incompleteness for the universal fragment~\cite{schroderdewittZXcalculusIncompleteQuantum2014}.
Moving forward, completeness for single qubit Clifford+T quantum mechanics was established in~\cite{backensZXcalculusCompleteSinglequbit2014}, and the completeness of the ZW-calculus was presented in~\cite{hadzihasanovicDiagrammaticAxiomatisationQubit2015} for qubits with integer coefficients and for the universal fragment in~\cite{hadzihasanovicAlgebraEntanglementGeometry2017}.
Further study of the ZX-calculus revealed the necessity of the \enquote{supplementarity} rule~\cite{perdrixSupplementarityNecessaryQuantum2016}.
Transferring the completeness proof of the ZW-calculus~\cite{hadzihasanovicDiagrammaticAxiomatisationQubit2015, hadzihasanovicAlgebraEntanglementGeometry2017} led to a surge of completeness results.
First, for the Clifford+T fragment in~\cite{jeandelCompleteAxiomatisationZXCalculus2018}, and then for the universal fragment in~\cite{ngUniversalCompletionZXcalculus2017, jeandelDiagrammaticReasoningClifford2018}.
Further improvements to the qubit ZX-calculus focused on minimizing its set of axioms, first in~\cite{backensMinimalStabilizerZXcalculus2020} and then in~\cite{vilmartNearMinimalAxiomatisationZXCalculus2019}.
This process was highly non-trivial, but now each rule is concise and intuitive.


Further to qubit calculi, the completeness of qudit calculi has also been extensively studied over the years.
The first such completeness result was that of the stabilizer fragment for qutrit ZX-calculus~\cite{wangQutritZXcalculusComplete2018}.
The stabilizer fragment of the ZX-calculus for all odd prime dimensions was shown to be complete in~\cite{boothCompleteZXcalculi2022}, and its rule-set has since been further reduced in~\cite{poorQupitStabiliserZXtravaganza2023}.
The first universal completeness of a graphical language beyond qubits was established in~\cite{poorCompletenessArbitraryFinite2023} for the arbitrary finite dimensions of ZXW-calculus.
Recently, the completeness of the qudit ZW-calculus has also been demonstrated in~\cite{devismeMinimalityFiniteDimensionalZWCalculi2024}, with a significantly reduced set of axioms.
Alongside the qudit ZW-calculus, the paper also introduces and proves the completeness of the finite-dimensional ZW-calculus~\cite{devismeMinimalityFiniteDimensionalZWCalculi2024}.
This result came shortly after the establishment of the qufinite ZXW-calculus and its completeness for finite-dimensional Hilbert spaces~\cite{wangCompletenessQufiniteZXW2024}.
However, the universal completeness of any higher-dimensional ZX-calculus has remained unproven.


In this paper, we define the finite-dimensional ZX-calculus and prove its completeness.
We begin by defining the generators of the calculus as the mixed-dimensional Z-spider~\cite{wangCompletenessQufiniteZXW2024} and the qudit X-spider, and then presenting its axioms.
Next, we recapitulate the axioms and definition of the finite-dimensional ZW-calculus~\cite{devismeMinimalityFiniteDimensionalZWCalculi2024}.
We then move on to proving the main result of the paper --- the completeness of the finite-dimensional ZX-calculus.
It is achieved by translating the generators of our language to the finite-dimensional ZW-calculus~\cite{devismeMinimalityFiniteDimensionalZWCalculi2024} and proving the invertibility of this translation.
Similar techniques of proving completeness have been employed before, first in~\cite{jeandelCompleteAxiomatisationZXCalculus2018}, and later in~\cite{ngUniversalCompletionZXcalculus2017, jeandelDiagrammaticReasoningClifford2018}.
By establishing completeness, we lay a solid foundation for the ZX-calculus as a versatile tool not only for quantum computation but also for various fields within finite-dimensional quantum theory.

\section{Finite-Dimensional ZX-calculus}\label{sec:zx-calculus}

In this section, we define our calculus starting with the generators, their interpretation, and lastly, the axiomatization that we later show to be complete.
Since any $1$-dimensional diagram is just the empty diagram, this paper only considers dimensions strictly bigger than $1$.

\subsection{Generators}\label{subsec:zx-generators}

We define the symmetric monoidal category $\ZXf$ with objects as lists of dimensions $\lst{d_i}_{i=1}^n$ where $d_i \in \N$ for all $0 < i \leq n$, such that $d_i \geq 2$, and morphisms generated by the following diagrams,
for any $a, b, n, m \in \N$, $0 \leq j < a$, and $\overrightarrow r \in \C^a$ such that $r_0 = 1$:
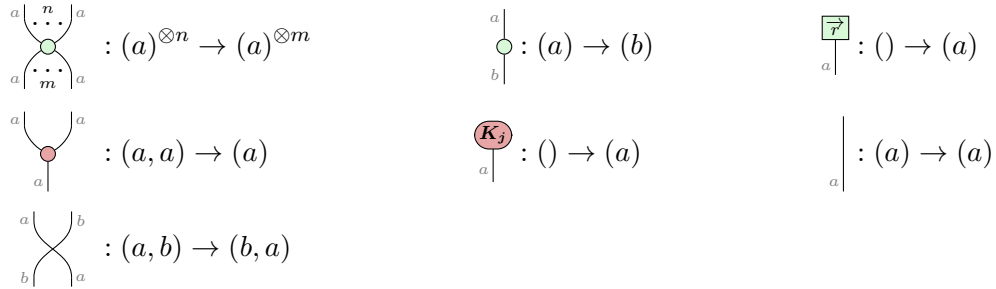
\begin{align*}
  \input{figures/generators/ZX/Z-mult.tikz} &: \lst{a}^{\otimes n} \to \lst{a}^{\otimes m} & \qquad
  \begin{tikzpicture}
	\begin{pgfonlayer}{nodelayer}
		\node [style=gbox] (0) at (0, 0.5) {$\overrightarrow{r}$};
		\node [style=none] (1) at (0, -0.75) {};
		\node [style=ZXdimension] (2) at (-0.25, -0.525) {$a$};
	\end{pgfonlayer}
	\begin{pgfonlayer}{edgelayer}
		\draw (1.center) to (0);
	\end{pgfonlayer}
\end{tikzpicture}
 &: \lst{} \to \lst{a} & \qquad
  \input{figures/generators/ZX/Z-embed.tikz} &: \lst{a} \to \lst{b} & & & \\
  \input{figures/generators/ZX/X-mult.tikz} &: \lst{a,a} \to \lst{a} & \qquad
  \begin{tikzpicture}
	\begin{pgfonlayer}{nodelayer}
		\node [style={rn_phase}] (0) at (0, 0.5) {$K_j$};
		\node [style=none] (1) at (0, -0.75) {};
		\node [style=ZXdimension] (2) at (-0.25, -0.45) {$a$};
	\end{pgfonlayer}
	\begin{pgfonlayer}{edgelayer}
		\draw (0) to (1.center);
	\end{pgfonlayer}
\end{tikzpicture}
 &: \lst{} \to \lst{a} &  \qquad
  \begin{tikzpicture}
	\begin{pgfonlayer}{nodelayer}
		\node [style=none] (0) at (0, 1) {};
		\node [style=none] (1) at (0, -1) {};
		\node [style=ZXdimension] (2) at (-0.25, -0.75) {$a$};
	\end{pgfonlayer}
	\begin{pgfonlayer}{edgelayer}
		\draw (0.center) to (1.center);
	\end{pgfonlayer}
\end{tikzpicture}
 &: \lst{a} \to \lst{a} & \qquad
  \input{figures/generators/ZX/braid.tikz} &: \lst{a,b} \to \lst{b,a}
\end{align*}
Diagrams are to be read top-to-bottom, as implied by the interpretation below.
Diagrams can be composed in two ways:
sequentially, by connecting input and output wires,
and in parallel, by placing them side-by-side.

We extend our language with the standard \emph{qudit Z-spider} and the \emph{mixed-dimensional Z-spider}~\cite{wangCompletenessQufiniteZXW2024},
given by the following compositions, respectively:
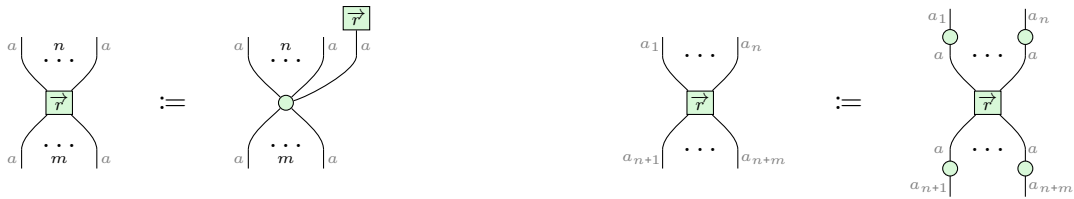
\[
  \input{figures/generators/ZX/qudit-Z-box.tikz}
  \quad \coloneqq \quad
  \input{figures/definitions/ZX/qudit-Z-box.tikz}
  \qquad \qquad
  \qquad \qquad
  \input{figures/generators/ZX/mixed-Z-box.tikz}
  \quad \coloneqq \quad
  \input{figures/definitions/ZX/mixed-Z-box.tikz}
\]
where $a = \min_{i = 0}^{n + m} a_i$ is the minimal dimension.
Furthermore, the phase-free Z-spider with arbitrary legs can express both the cap and cup as follows:
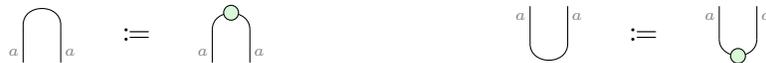
\[
  \input{figures/generators/ZX/cap.tikz}
  \quad \coloneqq \quad
  \input{figures/definitions/ZX/cap.tikz}
  \qquad \qquad
  \qquad \qquad
  \input{figures/generators/ZX/cup.tikz}
  \quad \coloneqq \quad
  \input{figures/definitions/ZX/cup.tikz}
\]
Caps and cups then allow us to construct the transposition of any diagram:
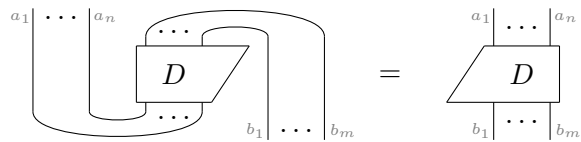
\[
  \begin{tikzpicture}
	\begin{pgfonlayer}{nodelayer}
		\node [style=none] (0) at (-6.75, 0.75) {};
		\node [style=none] (1) at (-3.75, 0.75) {};
		\node [style=none] (2) at (-4.75, -0.75) {};
		\node [style=none] (3) at (-6.75, -0.75) {};
		\node [style=none] (4) at (-5.75, 0) {$D$};
		\node [style=none] (5) at (-6.5, 0.75) {};
		\node [style=none] (6) at (-5, 0.75) {};
		\node [style=none] (7) at (-6.5, 1) {};
		\node [style=none] (8) at (-5, 1) {};
		\node [style=none] (9) at (-6.5, -0.75) {};
		\node [style=none] (10) at (-5, -0.75) {};
		\node [style=none] (11) at (-5.75, 1.125) {$\ldots$};
		\node [style=none] (12) at (-3.25, 1) {};
		\node [style=none] (13) at (-1.75, 1) {};
		\node [style=none] (14) at (-3.25, -1.75) {};
		\node [style=none] (15) at (-1.75, -1.75) {};
		\node [style=none] (16) at (-5, -1) {};
		\node [style=none] (17) at (-6.5, -1) {};
		\node [style=none] (18) at (-5.75, -1.15) {$\ldots$};
		\node [style=none] (19) at (-8, -1) {};
		\node [style=none] (20) at (-9.5, -1) {};
		\node [style=none] (21) at (-8, 1.75) {};
		\node [style=none] (22) at (-9.5, 1.75) {};
		\node [style=none] (27) at (-8.75, 1.5) {$\ldots$};
		\node [style=none] (28) at (-2.5, -1.5) {$\ldots$};
		\node [style=none] (29) at (-8, -0.5) {};
		\node [style=none] (30) at (0, 0) {$=$\strut};
		\node [style=none] (31) at (2.5, 0.75) {};
		\node [style=none] (32) at (4.5, 0.75) {};
		\node [style=none] (33) at (4.5, -0.75) {};
		\node [style=none] (34) at (1.5, -0.75) {};
		\node [style=none] (35) at (3.5, 0) {$D$};
		\node [style=none] (36) at (2.75, 0.75) {};
		\node [style=none] (37) at (4.25, 0.75) {};
		\node [style=none] (38) at (2.75, 1.75) {};
		\node [style=none] (39) at (4.25, 1.75) {};
		\node [style=none] (40) at (2.75, -1.75) {};
		\node [style=none] (41) at (2.75, -0.75) {};
		\node [style=none] (42) at (4.25, -1.75) {};
		\node [style=none] (43) at (4.25, -0.75) {};
		\node [style=none] (44) at (3.5, -1.25) {$\ldots$};
		\node [style=none] (45) at (3.5, 1.25) {$\ldots$};
		\node [style=ZXdimension] (50) at (-3.575, -1.5) {$b_1$};
		\node [style=ZXdimension] (51) at (-1.25, -1.5) {$b_m$};
		\node [style=ZXdimension] (52) at (-9.825, 1.5) {$a_1$};
		\node [style=ZXdimension] (53) at (-7.575, 1.5) {$a_n$};
		\node [style=ZXdimension] (56) at (2.425, -1.5) {$b_1$};
		\node [style=ZXdimension] (57) at (4.75, -1.5) {$b_m$};
		\node [style=ZXdimension] (58) at (2.425, 1.5) {$a_1$};
		\node [style=ZXdimension] (59) at (4.675, 1.5) {$a_n$};
	\end{pgfonlayer}
	\begin{pgfonlayer}{edgelayer}
		\draw (3.center) to (2.center);
		\draw (2.center) to (1.center);
		\draw (1.center) to (0.center);
		\draw (0.center) to (3.center);
		\draw (7.center) to (5.center);
		\draw (8.center) to (6.center);
		\draw [bend left=90, looseness=0.50] (8.center) to (12.center);
		\draw [bend left=90, looseness=0.50] (7.center) to (13.center);
		\draw (15.center) to (13.center);
		\draw (14.center) to (12.center);
		\draw [bend left=90, looseness=0.50] (17.center) to (19.center);
		\draw [bend left=90, looseness=0.50] (16.center) to (20.center);
		\draw (22.center) to (20.center);
		\draw (21.center) to (19.center);
		\draw (9.center) to (17.center);
		\draw (10.center) to (16.center);
		\draw (34.center) to (33.center);
		\draw (33.center) to (32.center);
		\draw (32.center) to (31.center);
		\draw (31.center) to (34.center);
		\draw (38.center) to (36.center);
		\draw (39.center) to (37.center);
		\draw (41.center) to (40.center);
		\draw (43.center) to (42.center);
	\end{pgfonlayer}
\end{tikzpicture}

\]
In particular, we obtain the transposition of all the generators.

With these diagrams, we can now define the \emph{X-spider} inductively, for any $m,n \in \N$:
\begin{equation*}
  \input{figures/generators/ZX/X-n-1.tikz}
  \coloneqq
  \input{figures/definitions/ZX/X-n-1.tikz}
  \qquad \qquad
  \input{figures/generators/ZX/X-spider-phasefree.tikz}
  \coloneqq
  \input{figures/definitions/ZX/X-spider-phasefree.tikz}
  \qquad \qquad
  \input{figures/generators/ZX/X-spider.tikz}
  \coloneqq
  \input{figures/definitions/ZX/X-spider.tikz}
\end{equation*}

\subsection{Interpretation}\label{subsec:interpretation}

The standard interpretation of a diagram in $\ZXf$ is a symmetric monoidal functor $\interp{\cdot}\strut : \ZXf \to \FHilb$ where $\FHilb$ is the category of finite-dimensional Hilbert spaces.
On objects, it is defined as $\interp{\lst{a_i}_{i=1}^n} = \C^{A}$ where $A = \prod_{i=1}^n a_i$, and on generators, it is given as follows:
\begin{align*}
  \input{figures/generators/ZX/Z-mult.tikz} \ &\interpTo\ \sum_{k=0}^{a-1} \ket{k}^{\otimes m}\bra{k}^{\otimes n} & \quad\quad
  \begin{tikzpicture}
	\begin{pgfonlayer}{nodelayer}
		\node [style=gbox] (0) at (0, 0.5) {$\overrightarrow{r}$};
		\node [style=none] (1) at (0, -0.75) {};
		\node [style=ZXdimension] (2) at (-0.25, -0.525) {$a$};
	\end{pgfonlayer}
	\begin{pgfonlayer}{edgelayer}
		\draw (1.center) to (0);
	\end{pgfonlayer}
\end{tikzpicture}
 \ &\interpTo\ \sum_{k=0}^{a-1} r_k \ket{k} & \quad\quad\quad
  \input{figures/generators/ZX/Z-embed.tikz}  \ &\interpTo\ \hspace{-10pt}\sum_{k=0}^{\min \{a, b\} - 1}\hspace{-10pt} \ket{k}\bra{k} & \\
  \input{figures/generators/ZX/X-mult.tikz} \ &\interpTo\ \sum_{k,\ell=0}^{a-1} \ket{k \plus \ell \hspace{-7pt} \mod a}\bra{k, \ell} & \quad
  \begin{tikzpicture}
	\begin{pgfonlayer}{nodelayer}
		\node [style=none] (1) at (0, -0.75) {};
		\node [style={rn_phase}] (2) at (0, 0.5) {$K_j$};
		\node [style=ZXdimension] (3) at (-0.25, -0.5) {$a$};
	\end{pgfonlayer}
	\begin{pgfonlayer}{edgelayer}
		\draw (2) to (1.center);
	\end{pgfonlayer}
\end{tikzpicture}
 \ &\interpTo\ \ket{a \minu j} & \quad\quad\quad
  \begin{tikzpicture}
	\begin{pgfonlayer}{nodelayer}
		\node [style=none] (0) at (0, 1) {};
		\node [style=none] (1) at (0, -1) {};
		\node [style=ZXdimension] (2) at (-0.25, -0.75) {$a$};
	\end{pgfonlayer}
	\begin{pgfonlayer}{edgelayer}
		\draw (0.center) to (1.center);
	\end{pgfonlayer}
\end{tikzpicture}
 \ &\interpTo\ \sum_{k=0}^{a-1} \ket{k}\bra{k} & \\
  \input{figures/generators/ZX/braid.tikz} \ &\interpTo\ \sum_{k=0}^{a-1} \sum_{l=0}^{b-1} \ket{\ell, k}\bra{k, \ell} & \quad\quad
  & &
  & &
\end{align*}
where $\overrightarrow r \in \C^a$ and $r_0 \coloneqq 1$.
Any other morphisms can be defined compositionally: $\interp{D_1 \otimes D_2 } = \interp{D_1} \otimes \interp{D_2}$, and $\interp{D_1 \circ D_2 } = \interp{D_1} \circ \interp{D_2}$.
By functoriality of the standard interpretation, the general spiders have the following correspondence in $\FHilb$:
\[
  \input{figures/generators/ZX/mixed-Z-box.tikz}
  \quad \interpTo \quad
  \sum_{j=0}^{a - 1} r_j
  \ket{j, \cdots, j} \bra{j, \cdots, j},
\]
where $a = \min\limits_{i = 0}^{n + m} a_i$, $\overrightarrow{r} = (r_1, \cdots, r_{a-1})$, and $r_0 = 1$; moreover, for $0 \leq j_p, k_q < a$,
\[
  \input{figures/generators/ZX/X-spider.tikz}
  \quad \interpTo \quad
  \sum_{
    \substack{
      i + j_1+\cdots+ j_m\\
      \equiv \, k_1+\cdots +k_n\! \Mod{a}
    }
  }
  \ket{j_1, \cdots, j_m}\bra{k_1, \cdots, k_n}.
\]

The mixed-dimensional Z-spider can have legs of varying dimensions which needs some further explanation:
In the qudit setting, a Z-spider behaves as the generalized Kronecker delta --- it ensures that the same basis state is present on each of its legs.
In our mixed-dimensional calculus, this behaviour is preserved by selecting the $k$-th standard basis on each leg for any $k$ less than the minimal dimension.
As such behaviour cannot be defined for the remaining basis states, we set their coefficients to zero.
In other words, these basis states are not included in the sum.

\subsection{Notations}

Here, we define some useful notations that we use throughout the paper.

\begin{itemize}
  \item The original green circle spider~\cite{coeckeInteractingQuantumObservables2011, ranchinDepictingQuditQuantum2014} can be defined using the Z box:
  \[
    \input{figures/definitions/ZX/circlegspiders1.tikz}
    \quad \overset{\interp{\cdot}}{\longmapsto} \quad
    \sum_{j=0}^{a-1} e^{i \alpha_j} \ket{j}^{\otimes m}\bra{j}^{\otimes n}
  \]
  where
  $a = \min\limits_{i = 0}^{n + m} a_i$,
  $\overrightarrow{\alpha} = (\alpha_1, \cdots, \alpha_{a-1})$,
  $\alpha_0 \coloneqq 0$,
  $e^{i\overrightarrow{\alpha}}=(e^{i\alpha_1}, \cdots, e^{i\alpha_{a-1}})$,\, and
  $\alpha_i \in [0, 2\pi)$.

  \item A \emph{multiplier}~\cite{caretteSZXCalculusScalableGraphical2019} labelled by $m$ corresponds to a Z- and X-spider connected with $m$ wires.
  Unlike in the qubit case, a green and a red spider can be connected with more than one wire.
  In fact, the Hopf law generalizes to $a$ connections in dimension $a$ (\cref{hopfditlm}), implying that the multiplier may be labelled modulo $a$.
  \begin{gather}
    \input{figures/generators/ZX/multiplier.tikz}
    \quad\coloneqq\quad
    \input{figures/definitions/ZX/multiplier.tikz}
    \qquad \qquad
    \input{figures/generators/ZX/multiplier-t.tikz}
    \quad\coloneqq\quad
    \input{figures/definitions/ZX/multiplier-t-1.tikz}
    \quad=\quad
    \input{figures/definitions/ZX/multiplier-t-2.tikz}
    \namedLabel{Mu}
  \end{gather}
  Then, the interpretation of the multiplier is given as follows:
  \[
    \input{figures/generators/ZX/multiplier.tikz}
    \quad\interpTo\quad
    \sum_{i=0}^{a-1} \ket{m \cdot i \mod a}\bra{i}
  \]

  \item We can define the \emph{dimension splitter} of the qufinite ZX-calculus~\cite{wangQufiniteZXcalculusUnified2022} as follows:
  \begin{gather}
    \input{figures/generators/ZX/dsplit.tikz}
    \quad\coloneqq\quad
    \input{figures/definitions/ZX/dsplit.tikz}
    \quad\interpTo\quad
    \sum_{i=0}^{a-1}\sum_{j=0}^{b-1}\ket{i,j}\bra{ib+j}
    \namedLabel{DD}
  \end{gather}
  Note that this definition matches that of~\cite[Axiom (DD)]{wangCompletenessQufiniteZXW2024}.

  \item Throughout this paper, we extensively use the vector $N = (1, \cdots, a-1)$ (where $a$ is the dimension), and elementwise functions on this vector.
  For example, $\sqrt{N!}$ corresponds to the vector $(\sqrt{1!}, \cdots, \sqrt{(a - 1)!})$ and $r^N$ refers to $(r^1, \cdots, r^{a-1})$.
  \item When two vectors are multiplied or added in the parameter of a Z-spider, we refer to elementwise multiplication or addition of the vectors.
  \item When the only non-zero element of the parameter in a Z-box is the first or last, we use the following shorthands:
  \[
    \input{figures/definitions/ZX/Z-spider-x-tilde-label.tikz}
    \qquad \qquad \qquad
    \input{figures/definitions/ZX/Z-spider-x-label.tikz}
  \]

  \item We use the notations $\overrightarrow{1_k} = (\overbrace{1, \cdots, 1}^{k - 1}, 0, \cdots, 0)$, $K_j=\left(j\frac{2\pi}{d}, 2j\frac{2\pi}{d}, \cdots, (d-1)j\frac{2\pi}{d}\right)$ in dimension $d$, and $\overrightarrow{1} = (1, \cdots, 1)$.

  \item The qudit generalization of the yellow Hadamard box is defined as follows, for $\omega = e^{i\frac{2\pi}{a}}$:
  \begin{gather}
    \scalebox{0.9}{\input{figures/definitions/ZX/hadamard.tikz}}
    \namedLabel{HD}
  \end{gather}
  \item The inverse of the Hadamard box, the yellow $H^\dagger$ box, is defined as follows:
  \begin{gather}
    \input{figures/definitions/ZX/h-dagger.tikz}
    \tag{H$\null^\dagger$}\label{rule:HDagger}\refstepcounter{equation}
  \end{gather}

  \item We use a yellow $D$ box to denote the \emph{dualiser} as defined in~\cite{coeckeInteractingQuantumObservables2011}:
  \begin{gather}
    \input{figures/definitions/ZX/dualiser.tikz}
    \quad \overset{\interp{\cdot}}{\longmapsto} \quad
    \sum_{i = 0}^{a \minu 1} \ket{i} \bra{a-i \Mod{a}}.
    \tag{Du}\label{rule:Du}\refstepcounter{equation}
  \end{gather}

  \item The Z-spider state with phase $\overrightarrow{1}$ and the X-spider state with phase $K_0$ are denoted with empty spiders:
  \[
    \input{figures/definitions/ZX/empty-Z-spider.tikz}
    \qquad \qquad \qquad
    \input{figures/definitions/ZX/empty-X-spider.tikz}\ ,
  \]
\end{itemize}

\subsection{Axiomatization}\label{subsec:axiomatisation}

In this section, we give a set of graphical rewrite rules.
First, an equation we assume for the Z-spider is called flexsymmetry~\cite{caretteWieldingZXcalculusFlexsymmetry2021}, enabling us to freely swap the legs of a Z-spider.
That is, for an arbitrary permutation of wires $\sigma$, the following holds:
\[
  \input{figures/axioms/ZX/Z-OCM.tikz}
\]
We write $\ZXf \vdash D_1 = D_2$, if we can use the rules of the calculus to turn $D_1$ into $D_2$.
Now, let us present the set of axioms for the calculus to perform purely diagrammatic reasoning:

\begin{multicols}{2}
  \allowdisplaybreaks
  \noindent
  \begin{gather}
    \input{figures/axioms/ZX/Z-fusion.tikz}
    \namedLabel{S1}
  \end{gather}
  where
  $A = \{a_t\}_{t = 0}^j$, $B = \{b_t\}_{t = 0}^\ell$, $C = \{c_t\}_{t = 0}^p$,\\
  $M = \min(A {\cup} B {\cup} C)$,
  $m = \min(A {\cup} C)$,
  $n = \min(B {\cup} C)$,
  $v = \min(A {\cup} B)$,
  $w = \min(C)$,
  $\overrightarrow{r}=(r_1, \dotsc, r_{m-1})$,
  $\overrightarrow{s}=(s_1, \dotsc, s_{n-1})$, and
  $\overrightarrow{rs'}=(r_1 s_1,~ \dotsc,~ r_{M-1} s_{M-1}, \underbrace{0,~ \dotsc,~ 0}_{\max(v \minu w,\, 0)})$.
  \vspace*{-6pt}
  \begin{gather}
    \input{figures/axioms/ZX/color.tikz}
    \namedLabel{HZ}
  \end{gather}
  where $0 \leq j < a$ and $u_{m,n} = a^{\frac{m+n-2}{2}}$
  \begin{gather}
    \input{figures/axioms/ZX/bialgebra.tikz}
    \namedLabel{B2} \\
    \input{figures/axioms/ZX/p1.tikz}
    \namedLabel{P1} \\
    \input{figures/axioms/ZX/id.tikz}
    \namedLabel{S2} \\
    \input{figures/axioms/ZX/copy.tikz}
    \namedLabel{K0}
  \end{gather}
  where $N = \min\{a, b, c\}$.
  \begin{gather}
    \input{figures/axioms/ZX/empty.tikz}
    \namedLabel{Ept} \\
    \input{figures/axioms/ZX/d1.tikz}
    \namedLabel{D1}
  \end{gather}
  where $\overleftarrow{r} = (r_{a \minu 1}, \cdots, r_{1})$
  \begin{gather}
    \input{figures/axioms/ZX/piCommute2.tikz}
    \namedLabel{K2}
  \end{gather}
  where $\displaystyle \protect{\hat{k}(\overleftarrow{r})}=\left(\frac{r_{1-j}}{r_{a-j}}, \dotsc, \frac{r_{a-1-j}}{r_{a-j}}\right)$
  \begin{gather}
    \input{figures/axioms/ZX/hopf.tikz}
    \namedLabel{HP} \\
    \input{figures/axioms/ZX/X-spider-mod.tikz}
    \namedLabel{XM} \\
    \input{figures/axioms/ZX/dimInc.tikz}
    \namedLabel{DA}
  \end{gather}
  where $k \in \N$.
  \begin{gather}
    \input{figures/axioms/ZX/phaseAddQudit.tikz}
    \namedLabel{PA}
  \end{gather}
  where $\displaystyle r_k=\!\sum_{\substack{i = 0}}^{a \minu 1}  p_i q_{k \minu i \ \mathrm{mod}\ a}\!$ is the $k^{\text{th}\!}$ entry of $\!\overrightarrow{r}\!$.
  \begin{gather}
    \input{figures/axioms/ZX/zState.tikz}
    \namedLabel{ZNF} \\
    \input{figures/axioms/ZX/phaseCopy.tikz}
    \namedLabel{PC}
  \end{gather}
  where $\overrightarrow q$ is a vector such that for all $0 \leq i < b$ and $0 \leq j < c$ we have $p_i p_j = q_{i + j \! \Mod{a}}$.
  \begin{gather}
    \scalebox{0.85}{\input{figures/axioms/ZX/ugly.tikz}}
    \namedLabel{WW}
  \end{gather}
  where $c\geq\min(\sum a_i, \sum b_i)$, $\ell_{ij}=\min(a_i,b_j)$, $A_i = \sum_k{\ell_{i, k}}$, $B_i = \sum_k{\ell_{k, i}}$, $A = \sum_k a_k$, and $B = \sum_k b_k$.
\end{multicols}

\subsubsection{Explanation of Axioms}

While many of the axioms have already been presented in other papers, some appear here for the first time;
therefore, this section provides a brief explanation of the newly introduced axioms.

\begin{itemize}
  \item \eqref{rule:S1} is the generalization of the qudit fusion rule for mixed-dimensional Z-spiders. The phases are multiplied but the vector needs to be cut off at the minimal dimension.
  \item \eqref{rule:D1} generalizes \zxw{(D1)} for an arbitrary number of legs, making \zxw{(H1)} redundant.
  \item \eqref{rule:K2} generalizes \zxw{(K2)} and \zxw{(K1)}.
  \item \eqref{rule:K0} is the mixed-dimensional copy rule.
  The scalar is $0$, if the basis state $\ket{j}$ exceeds the capacity of the output dimension, and $1$ otherwise.
  \item \eqref{rule:PA}, standing for \emph{Phase Addition}, expresses how an X-spider sums up two Z-box states.
  \item \eqref{rule:PC}, standing for \emph{Phase Commute}, commutes a class of Z-boxes through an X-spider.
  \item \eqref{rule:HP} is the translation of the (h) rule of the finite-dimensional ZW-calculus.
  \item The diagrams in~\eqref{rule:DA} implement the scalar-less W-node, with the rule stating that we can change the internal dimension as long as it is sufficiently high.
  \item \eqref{rule:ZNF} maps an arbitrary Z-spider into its normal form.
  \item \eqref{rule:XM} asserts that an X-spider can also be implemented as a scalar-less W-node composed with a modulo box; see \cref{subsec:iXW} for further explanation.
  \item \eqref{rule:WW} directly translates the ($b_2$) rule of the finite-dimensional ZW-calculus.
\end{itemize}

\section{Finite-Dimensional ZW-calculus\texorpdfstring{~\cite{devismeMinimalityFiniteDimensionalZWCalculi2024}}{}}\label{sec:zw-calculus}

This section recapitulates the finite-dimensional ZW-calculus, in accordance with~\cite{devismeMinimalityFiniteDimensionalZWCalculi2024}, discussing its generators with their interpretations and presenting the complete axiomatization of the calculus.\footnote{This paper builds on the equational theory presented in Ref.~\cite{devismeMinimalityFiniteDimensionalZWCalculi2024};
however, the referenced paper has since been updated and published, with a slightly modified set of rewrite rules; see Ref.~\cite{devismeMinimalityFiniteDimensional2025}.}

\subsection{Generators}

We define the symmetric monoidal category $\ZWf$ with objects as lists of dimensions $\lst{d_i}_{i=1}^n$ where $d_i \in \N$ for all $0 < i \leq n$, and morphisms generated by the following diagrams,
for any $a, b, b_j, n, m \in \N$, $0 < j \leq n$, $0 < k \leq a$, and $r \in \C$:
\begin{align*}
  \input{figures/generators/ZW/Z-spider.tikz} &: \lst{a}^{\otimes n} \to \lst{a}^{\otimes m} & \quad\quad\quad
  \input{figures/generators/ZW/W-spider.tikz} &: \lst{a} \to \lst{b_i}_{i=1}^n & \quad\quad\quad
  \begin{tikzpicture}
	\begin{pgfonlayer}{nodelayer}
		\node [style=ket] (0) at (0, 0.5) {$k$};
		\node [style=none] (1) at (0, -0.5) {};
		\node [style=rZWdimension] (2) at (0, -0.1) {$a$\!};
	\end{pgfonlayer}
	\begin{pgfonlayer}{edgelayer}
		\draw (1.center) to (0);
	\end{pgfonlayer}
\end{tikzpicture}
 &: \lst{} \to \lst{a} & \\
  \input{figures/generators/ZW/swap-annot.tikz} &: \lst{a,b} \to \lst{b,a} & \quad\quad\quad
  \input{figures/generators/ZW/cap-annot.tikz} &: \lst{} \to \lst{a,a} & \quad\quad\quad
  \input{figures/generators/ZW/cup-annot.tikz} &: \lst{a,a} \to \lst{} & \\
  r &: \lst{ } \to \lst{ } & \quad\quad\quad
  \input{figures/generators/ZW/id-annot.tikz} &: \lst{a} \to \lst{a} & \quad\quad\quad
  & &
\end{align*}
Compositions are given in the usual way.

\subsection{Interpretation}

The interpretation of a diagram in $\ZWf$ is a symmetric monoidal functor $\interp{\cdot}\strut : \ZWf \to \FHilb$.
On objects, it is defined as $\interp{\lst{a_i}_{i=1}^n} = \C^{A}$ where $A = \prod_{i=1}^n (a_i+1)$.
The interpretation of the Z-spider is as follows, for any $r \in \C$:
\[
  \input{figures/generators/ZW/Z-spider.tikz}
  \quad \interpTo \quad
  \sum_{k=0}^a r^k \sqrt{k!}^{n+m-2}\ketbra{k^m}{k^n}
\]
The W-node and its interpretation are given as follows:
\[
  \input{figures/generators/ZW/W-spider.tikz}
  \quad \interpTo \quad
  \hspace*{-1em}\sum_{\substack{0\leq k_i \leq b_i\\ k_1{+}...{+}k_n \leq a}}
  \sqrt{\binom{k_1{+}...{+}k_n}{k_1,...,k_n}}\ketbra{k_1,...,k_n}{k_1{+}...{+}k_n}
\]
where $a \geq \max_{i = 1}^{n} b_i$.
The interpretations of the remaining generators are as follows:
\begin{align*}
  \begin{tikzpicture}
	\begin{pgfonlayer}{nodelayer}
		\node [style=ket] (0) at (0, 0.5) {$k$};
		\node [style=none] (1) at (0, -0.5) {};
		\node [style=rZWdimension] (2) at (0, -0.1) {$a$\!};
	\end{pgfonlayer}
	\begin{pgfonlayer}{edgelayer}
		\draw (1.center) to (0);
	\end{pgfonlayer}
\end{tikzpicture}
 \ \interpTo \ &
  \begin{cases}
    \sqrt{k!} \ket k  & \text{if } 0 < k \leq a \\
    \vec 0            & \text{otherwise}
  \end{cases} &
  r \ \interpTo \ & r & \quad\quad\quad
  \input{figures/generators/ZW/id-annot.tikz} \ \interpTo \ & \sum_{k=0}^{a} \ket{k}\bra{k} & \\
  \input{figures/generators/ZW/swap-annot.tikz} \ \interpTo \ & \sum_{k=0}^{a} \sum_{l=0}^{b} \ket{\ell, k}\bra{k, \ell} &  \quad\quad\quad
  \input{figures/generators/ZW/cap-annot.tikz} \ \interpTo \ & \sum_{k=0}^{a} \ket{k, k} & \quad\quad\quad
  \input{figures/generators/ZW/cup-annot.tikz} \ \interpTo \ & \sum_{k=0}^{a} \bra{k, k} &
\end{align*}

It is worth pointing out that the dimensions of the wire in $\ZXf$ and $\ZWf$ are not the same.
In the ZX-calculus $d$ is used to indicate that the wire carries a $d$-dimensional qudit while in the case of the ZW-calculus it means a $(d+1)$-dimensional qudit.
On ZX-diagrams, dimension is marked with grey text by the wire, whereas text in light grey bubble next to the wire indicates the dimension of the ZW-calculus.

\subsection{Axioms}\label{subsec:zw-axioms}

The equational theory $\ZWf$ for the finite-dimensional ZW-calculus is given as follows:\\
\boxed{
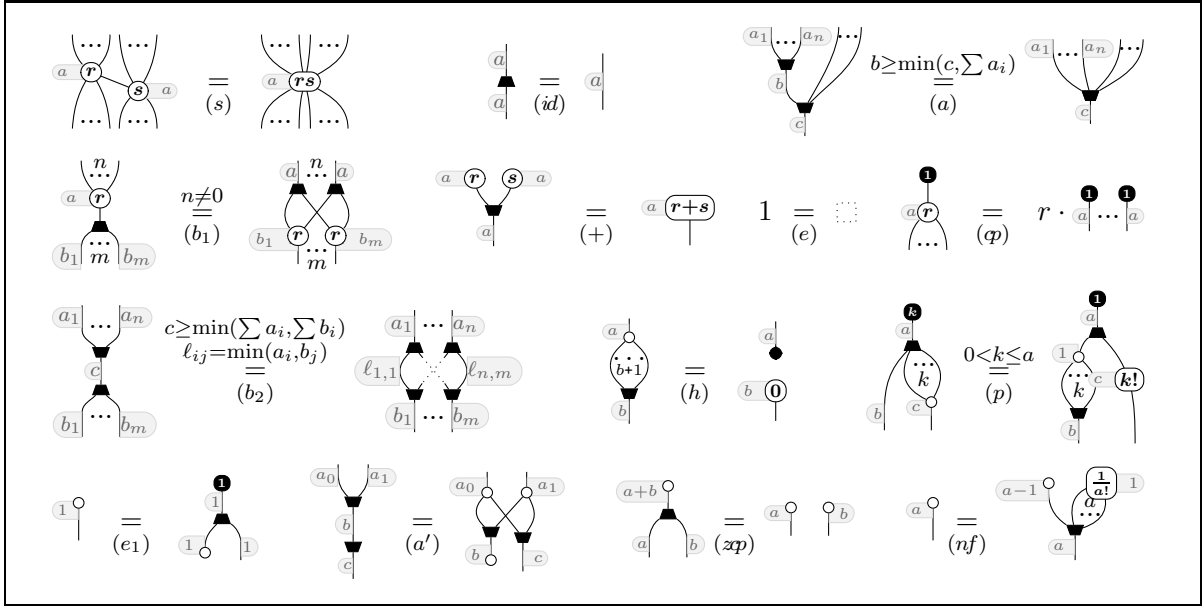
  \begin{minipage}{0.98\textwidth}
    \medskip
    \hspace*{1em}
    $\label{ax:Z-spider-qf}
    \input{figures/axioms/ZW/Z-spider-rule-annot-00.tikz}
    \eq{(s)}\input{figures/axioms/ZW/Z-spider-rule-annot-01.tikz}$
    \hfill
    $\label{ax:id-qf}
    \input{figures/axioms/ZW/Z-id-annot-02.tikz}
    \eq[~]{(i\hspace*{-1pt}d)}\input{figures/axioms/ZW/Z-id-annot-01.tikz}$
    \hfill
    $\label{ax:W-assoc-qf}
    \input{figures/axioms/ZW/W-assoc-annot-00.tikz}
    \overset{b\geq\min(c,\sum a_i)}{\eq{(a)}}\input{figures/axioms/ZW/W-assoc-annot-01.tikz}$
    \hspace*{1em}
    \medskip

    \hspace*{1em}
    $\label{ax:bialgebra-Z-W-qf}
    \input{figures/axioms/ZW/Z-W-bialgebra-annot-00.tikz}
    \overset{n\neq0}{\eq{(b_1)}}\input{figures/axioms/ZW/Z-W-bialgebra-annot-01.tikz}$
    \hfill
    $\label{ax:sum-qf}
    \input{figures/axioms/ZW/sum-annot-00.tikz}
    \eq{(+)}\begin{tikzpicture}
	\begin{pgfonlayer}{nodelayer}
		\node [style=rZWdimension] (12)  at (-0.124, 0.124) {$a$~~.};
		\node [style=white phase dot] (5)  at (0.124, 0.124) {$r{+}s$};
		\node [style=none] (6)  at (0.124, -0.876) {};
		\node [style=none] (7)  at (0.124, 0.876) {};
	\end{pgfonlayer}
	\begin{pgfonlayer}{edgelayer}
		\draw (5) to (6.center);
	\end{pgfonlayer}
\end{tikzpicture}$
    \hfill
    $\label{ax:one-qf}
    1\eq[~]{(e)}\input{figures/axioms/ZW/empty.tikz}$
    \hfill
    $\label{ax:scalar-qf}
    \input{figures/axioms/ZW/copy-1-annot-00.tikz}
    \eq{(c\!p)} r\cdot \input{figures/axioms/ZW/copy-1-annot-01.tikz}$
    \hspace*{1em}
    \medskip

    \hspace*{1em}
    $\label{ax:bialgebra-W-qf}
    \input{figures/axioms/ZW/W-bialgebra-annot-00.tikz}
    \overset{\substack{c\geq\min(\sum a_i, \sum b_i)\\\ell_{ij}=\min(a_i,b_j)}}{\eq{(b_2)}}\input{figures/axioms/ZW/W-bialgebra-annot-01.tikz}$
    \hfill
    $\label{ax:hopf-qf}
    \input{figures/axioms/ZW/Hopf-annot-00.tikz}
    \eq{(h)}\input{figures/axioms/ZW/Hopf-annot-01.tikz}$
    \hfill
    $\label{ax:partition-qf}
    \input{figures/axioms/ZW/ket-1-ket-k-annot-00.tikz}
    ~\overset{0< k\leq a}{\eq{(p)}}~\input{figures/axioms/ZW/ket-1-ket-k-annot-01.tikz}$
    \hspace*{1em}
    \medskip

    \hspace*{1em}
    $\label{ax:e1-qf}
    \begin{tikzpicture}
	\begin{pgfonlayer}{nodelayer}
		\node [style=rZWdimension] (14)  at (0.0, 0.3) {$1~$};
		\node [style=white dot] (6)  at (0.0, 0.5) {};
		\node [style=none] (7)  at (0.0, -0.5) {};
	\end{pgfonlayer}
	\begin{pgfonlayer}{edgelayer}
		\draw (7.center) to (6);
	\end{pgfonlayer}
\end{tikzpicture}
    \eq{(e_1)}\input{figures/axioms/ZW/e1-NF-annot-01.tikz}$
    \hfill
    $\label{ax:W-assoc-2-qf}
    \input{figures/axioms/ZW/W-assoc-annot2-00.tikz}
    \eq[]{(a')}\input{figures/axioms/ZW/W-assoc-annot2-01.tikz}$
    \hfill
    $\label{ax:Z-copy-qf}
    \input{figures/axioms/ZW/Z-copy-annot-00.tikz}
    \eq[]{(z\!c\!p)}\input{figures/axioms/ZW/Z-copy-annot-01.tikz}$
    \hfill
    $\label{ax:Z-normal-form-qf}
    \begin{tikzpicture}
	\begin{pgfonlayer}{nodelayer}
		\node [style=rZWdimension] (11)  at (0.0, 0.3) {$a~$};
		\node [style=white dot] (1)  at (0.0, 0.5) {};
		\node [style=none] (2)  at (0.0, -0.5) {};
	\end{pgfonlayer}
	\begin{pgfonlayer}{edgelayer}
		\draw (2.center) to (1);
	\end{pgfonlayer}
\end{tikzpicture}
    \eq[]{(n\!f)}\input{figures/axioms/ZW/Z-NF-annot-01.tikz}$
    \hspace*{1em}
    \medskip
  \end{minipage}
}

\begin{proposition}[Completeness of ZW-calculus]
\label{prop:ZWcomplete}
For any two $\ZWf$-diagrams of the same type $D_1: A \to B$ and $D_2: A \to B$,
\[
  \interp{D_1} = \interp{D_2} \iff \ZWf \vdash D_1 = D_2
\]
\end{proposition}
\noindent The proof is the content of~\cite[Theorem 3]{devismeMinimalityFiniteDimensionalZWCalculi2024}.

\begin{remark}
  While in $\ZXf$ all wire dimensions must be greater than or equal to $2$, the
  finite-dimensional ZW-calculus allows wires of dimension $1$ (labelled by $0$).
  Therefore, we require the completeness of the finite-dimensional ZW-calculus restricted to having no $0$-labeled wires.
  This variant of the ZW completeness is still complete as the proof does not use $0$-labeled wires.
\end{remark}

\section{Completeness from Translation}\label{sec:translation-completeness}

In this section, we prove the completeness of the finite-dimensional ZX-calculus.
Our proof strategy is based on the one outlined in~\cite{jeandelCompleteAxiomatisationZXCalculus2018, ngUniversalCompletionZXcalculus2017}.
The idea of the proof is to define a translation from and to a complete language --- the finite-dimensional ZW-calculus in our case.
Then, if one can prove certain properties of these translation functors, such as soundness and invertibility, then we obtain completeness.

\subsection{ZX-to-ZW Translation}\label{subsec:iXW}

Here, we define the \emph{ZX-to-ZW translation functor}, $\iXW : \ZXf \to \ZWf$.  On objects, we have $\iXW[\lst{a_i}_{i=1}^n] = \lst{a_i-1}_{i=1}^n$. On morphisms, we show how it maps the generators of $\ZXf$ to the generators of $\ZWf$; the other morphisms can be defined compositionally:  $\iXW[D_1 \otimes D_2 ] =\iXW[D_1] \otimes \iXW[D_2]$, and  $\iXW[D_1 \circ D_2 ] =\iXW[D_1] \circ \iXW[D_2]$.
We first show how the generators defining a general mixed-dimensional Z-spider are translated,
starting with the Z-spider with a single output and the embedding:
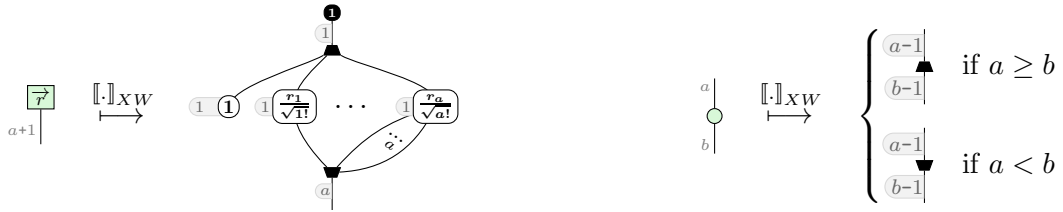
\[
  \begin{tikzpicture}
	\begin{pgfonlayer}{nodelayer}
		\node [style=gbox] (0) at (0, 0.5) {$\overrightarrow{r}$};
		\node [style=none] (1) at (0, -0.75) {};
		\node [style=ZXdimension] (2) at (-0.5, -0.375) {$a \plus 1$};
	\end{pgfonlayer}
	\begin{pgfonlayer}{edgelayer}
		\draw (1.center) to (0);
	\end{pgfonlayer}
\end{tikzpicture}

  \quad\ZXtoZW\quad
  \scalebox{0.9}{\input{figures/translation/from-ZW/Z-box-state.tikz}}
  \qquad \qquad
  \qquad \qquad
  \input{figures/generators/ZX/Z-embed.tikz}
  \quad\ZXtoZW\quad
  \begin{cases}
    \scalebox{0.9}{\input{figures/translation/from-ZW/Z-embed-1.tikz}}& \text{if } a \geq b\\[15pt]
    \scalebox{0.9}{\input{figures/translation/from-ZW/Z-embed-2.tikz}}& \text{if } a < b
  \end{cases}
\]
The former interpretation is based on the normal form of the qudit ZW-calculus~\cite[Definition 3]{devismeMinimalityFiniteDimensionalZWCalculi2024}, and the embedding follows from comparing the interpretations.
Since the Z-spiders of the two calculi match up to a $C = \left(\frac{1}{\sqrt{N!}}\right)^{m+n-1}$ vector parameter, its translation is as follows:
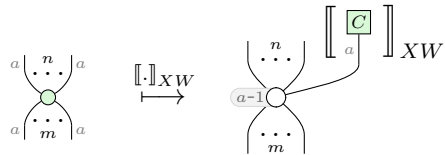
\[
  \input{figures/generators/ZX/Z-mult.tikz}
  \quad\ZXtoZW\quad
  \input{figures/translation/from-ZW/phase-Z-box.tikz}
\]

Now, we focus on expressing generators that include X-spiders.
The translation of computational basis states follows from the interpretations.
To express the X-spider in the finite-dimensional ZW-calculus, we first point out that the interpretation of the W- and X-spider are closely related.
Other than the scalar factors, the only difference is that the W-spider results in sums of the input states that are smaller than the output dimension while the X-spider sums elements modulo $a$:
\begin{align*}
  \input{figures/generators/ZW/W-spider-2-1.tikz}
  \quad&\interpTo\quad
  \ket{k + \ell}\bra{k, \ell} \quad\text{for}\quad k, \ell < a \ \text{and} \ k + \ell < a, \\
  \input{figures/generators/ZX/X-spider-2-1.tikz}
  \quad&\interpTo\quad
  \ket{k + \ell \!\!\mod a}\bra{k, \ell} \quad\text{for}\quad k, \ell < a
\end{align*}
Setting the output dimension of the W-spider to be $2a$, and applying a modulo $a$ gate,
we obtain the translation of the X spider, that is,
\begin{align*}
  \ket{j \!\!\mod a}\bra{j}\ket{k + \ell}\bra{k, \ell} \ \ \text{for}\ \ k, \ell < a \ \land\ k + \ell < 2a,
  \quad\Longleftrightarrow\quad
  \ket{k + \ell \!\!\mod a}\bra{k, \ell} \ \ \text{for}\ \ k, \ell < a,
\end{align*}
Diagrammatically, these translations are given as follows:
\[
  \begin{tikzpicture}
	\begin{pgfonlayer}{nodelayer}
		\node [style={rn_phase}] (0) at (0, 0.5) {$K_j$};
		\node [style=none] (1) at (0, -0.75) {};
		\node [style=ZXdimension] (2) at (-0.25, -0.45) {$a$};
	\end{pgfonlayer}
	\begin{pgfonlayer}{edgelayer}
		\draw (0) to (1.center);
	\end{pgfonlayer}
\end{tikzpicture}

  \quad\ZXtoZW\quad
  \input{figures/translation/from-ZW/ket.tikz}
  \qquad \quad
  \input{figures/generators/ZX/X-spider-2-1.tikz}
  \quad\ZXtoZW\quad
  \scalebox{0.9}{\input{figures/translation/from-ZW/X-spider-2-1.tikz}}
\]
where the modulo $a$ gadget is given as follows:
\[
  \input{figures/definitions/ZX/mod-d.tikz}
  \quad\ZXtoZW\quad
  \input{figures/generators/ZW/mod-d.tikz}
  \quad\coloneqq\quad
  \begin{tikzpicture}
	\begin{pgfonlayer}{nodelayer}
		\node [style=W-1-n] (0) at (0, 0.5) {};
		\node [style=none] (1) at (-0.25, -0.375) {\scalebox{.7}{$\ldots$}};
		\node [style=none] (2) at (0.75, -0.5) {};
		\node [style=none] (3) at (0.75, -3.375) {};
		\node [style=none] (4) at (0, 3.25) {};
		\node [style=ZXdimension] (8) at (-0.25, -0.125) {$a$};
		\node [style=rZWdimension] (9) at (-0.5, -1) {$1$~~};
		\node [style=white dot] (10) at (-0.5, -1) {};
		\node [style=gbox] (11) at (0, 2.25) {\tiny $\frac{1}{\sNf}$};
		\node [style=gbox] (12) at (0.75, -2.25) {$\sNf$};
		\node [style=none] (14) at (-1.1, 2.25) {$\Big\llbracket$};
		\node [style=none] (15) at (1.65, 2.25) {$\Big\rrbracket_{XW\!}$};
		\node [style=none] (16) at (-0.35, -2.25) {$\Big\llbracket$};
		\node [style=none] (17) at (2.4, -2.25) {$\Big\rrbracket_{XW\!}$};
		\node [style=rZWdimension] (18) at (0, 1) {$2a \minu 1$~};
		\node [style=lZWdimension] (19) at (0.75, -0.999) {$~a \minu 1$};
	\end{pgfonlayer}
	\begin{pgfonlayer}{edgelayer}
		\draw (3.center) to (2.center);
		\draw [bend left=45] (10) to (0);
		\draw [bend right=45] (10) to (0);
		\draw [in=90, out=-30, looseness=0.50] (0) to (2.center);
		\draw (4.center) to (0);
	\end{pgfonlayer}
\end{tikzpicture}

\]
The translations of the remaining generators, the identity and the swap, are trivial.
\begin{lemma}
  \label{lem:ZW-sound}
  The ZX-to-ZW translation functor $\iXW$ preserves the semantics, that is, $\interp{\iXW[D]} = \interp{D}$ for any ZX-diagram $D$.
\end{lemma}
\noindent To verify this lemma, it suffices to apply the computational basis states to both the original ZX-diagram $D$ and its translated ZW-diagram $\iXW[D]$.
By evaluating these diagrams under all possible computational basis inputs, we can verify their equivalence under interpretations.

\subsection{ZW-to-ZX Translation}

Here, we define the \emph{ZW-to-ZX translation functor}, $\iWX : \ZWf \to \ZXf$.  On objects, we have $\iWX[\lst{a_i}_{i=1}^n] = \lst{a_i+1}_{i=1}^n$. On morphisms, we show how it maps the generators of $\ZWf$ to the generators of $\ZXf$; the other morphisms can be defined compositionally:  $\iWX[D_1 \otimes D_2 ] =\iWX[D_1] \otimes \iWX[D_2]$, and  $\iWX[D_1 \circ D_2 ] =\iWX[D_1] \circ \iWX[D_2]$.
 For $B = \sum_{i = 1}^{n} b_i$,
\[
  \input{figures/generators/ZW/Z-spider.tikz}
  \quad\ZWtoZX\quad
  \begin{tikzpicture}
	\begin{pgfonlayer}{nodelayer}
		\node [style=none] (61) at (1, 1.25) {};
		\node [style=none] (62) at (1, 1.75) {};
		\node [style=none] (63) at (0, 1.125) {$\ldots$};
		\node [style=none] (64) at (-1, 1.25) {};
		\node [style=none] (65) at (-1, 1.75) {};
		\node [style=none] (66) at (1, -1.25) {};
		\node [style=none] (67) at (1, -1.75) {};
		\node [style=none] (68) at (0, -1.125) {$\ldots$};
		\node [style=gbox] (69) at (0, 0) {\tiny $r^N \sNf^{n \plus m \minu 2}$};
		\node [style=none] (70) at (-1, -1.25) {};
		\node [style=none] (71) at (-1, -1.75) {};
		\node [style=ZXdimension] (72) at (-1.425, -1.5) {$a \plus 1$};
		\node [style=ZXdimension] (73) at (1.45, -1.5) {$a \plus 1$};
		\node [style=ZXdimension] (74) at (1.45, 1.5) {$a \plus 1$};
		\node [style=ZXdimension] (75) at (-1.425, 1.5) {$a \plus 1$};
		\node [style=wire label] (76) at (0, 1.5) {$n$};
		\node [style=wire label] (77) at (0, -1.5) {$m$};
	\end{pgfonlayer}
	\begin{pgfonlayer}{edgelayer}
		\draw (62.center) to (61.center);
		\draw (65.center) to (64.center);
		\draw (67.center) to (66.center);
		\draw [style=none, in=90, out=-135, looseness=0.75] (69) to (70.center);
		\draw (71.center) to (70.center);
		\draw [style=none, in=90, out=-45, looseness=0.75] (69) to (66.center);
		\draw [style=none, in=-90, out=135, looseness=0.75] (69) to (64.center);
		\draw [style=none, in=-90, out=45, looseness=0.75] (69) to (61.center);
	\end{pgfonlayer}
\end{tikzpicture}

  \qquad\qquad\qquad
  \input{figures/generators/ZW/W-spider.tikz}
  \quad\ZWtoZX\quad
  \begin{tikzpicture}
	\begin{pgfonlayer}{nodelayer}
		\node [style=rn] (9) at (0, 0) {};
		\node [style=gbox] (10) at (-1, -1.25) {\tiny $\frac{1}{\sNf}$};
		\node [style=gbox] (11) at (1, -1.25) {\tiny $\frac{1}{\sNf}$};
		\node [style=none] (12) at (-1, -2.5) {};
		\node [style=none] (13) at (1, -2.5) {};
		\node [style=ZXdimension] (14) at (-1.625, -2.125) {$b_1 \plus 1$};
		\node [style=ZXdimension] (15) at (1.625, -2.125) {$b_n \plus 1$};
		\node [style=gbox] (16) at (0, 1.25) {$\sNf$};
		\node [style=none] (17) at (0, 2.375) {};
		\node [style=ZXdimension] (18) at (-0.5, 2.125) {$a \plus 1$};
		\node [style=ZXdimension] (20) at (-0.75, 0.25) {$B \plus 1$};
		\node [style=none] (22) at (0, -2.225) {$\ldots$};
	\end{pgfonlayer}
	\begin{pgfonlayer}{edgelayer}
		\draw [in=-150, out=90, looseness=0.75] (10) to (9);
		\draw [in=90, out=-30, looseness=0.75] (9) to (11);
		\draw (13.center) to (11);
		\draw (10) to (12.center);
		\draw (16) to (9);
		\draw (17.center) to (16);
	\end{pgfonlayer}
\end{tikzpicture}

\]
The translation of the Z-spider follows directly from its interpretation.
The intuition behind the translation of the W-node is similar to that used for translating the X-spider from ZX to ZW\@.
Up to scalars, both nodes sum up basis states, but the X-spider also takes the result modulo the dimension.
Embedding the X-spider in a sufficiently high dimension, we can ensure that the sum of the basis states is always less than the dimension and thus the modulo is never taken.
Once we have this, the mixed-dimensional Z-spider can remove sums that are larger than the output dimension of the W-node resulting in the same interpretation.

The remaining translations follow directly from the interpretation:
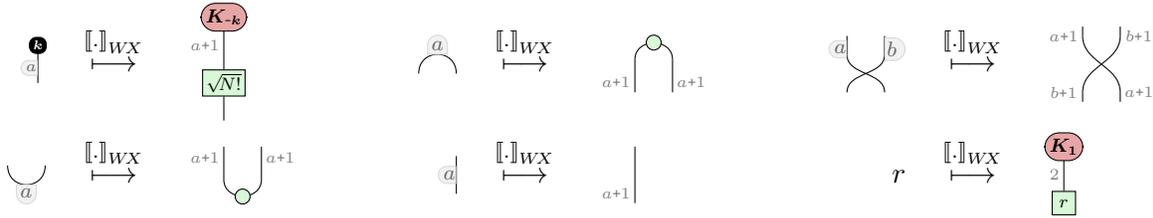
\begin{align*}
  \begin{tikzpicture}
	\begin{pgfonlayer}{nodelayer}
		\node [style=ket] (0) at (0, 0.5) {$k$};
		\node [style=none] (1) at (0, -0.5) {};
		\node [style=rZWdimension] (2) at (0, -0.1) {$a$\!};
	\end{pgfonlayer}
	\begin{pgfonlayer}{edgelayer}
		\draw (1.center) to (0);
	\end{pgfonlayer}
\end{tikzpicture}
 &\quad\ZWtoZX\quad \begin{tikzpicture}
	\begin{pgfonlayer}{nodelayer}
		\node [style={rn_phase}] (4) at (0, 1) {$K_{\minu k}$};
		\node [style=ZXdimension] (6) at (-0.5, 0.25) {$a \plus 1$};
		\node [style=gbox] (7) at (0, -0.5) {$\sNf$};
		\node [style=none] (8) at (0, -1.5) {};
	\end{pgfonlayer}
	\begin{pgfonlayer}{edgelayer}
		\draw (7) to (4);
		\draw (8.center) to (7);
	\end{pgfonlayer}
\end{tikzpicture}
 & \quad\quad\quad
  \input{figures/generators/ZW/cap-annot.tikz} &\quad\ZWtoZX\quad \begin{tikzpicture}
	\begin{pgfonlayer}{nodelayer}
		\node [style=none] (0) at (-0.5, 0.125) {};
		\node [style=none] (1) at (0.5, 0.125) {};
		\node [style=none] (2) at (-0.5, -0.75) {};
		\node [style=none] (3) at (0.5, -0.75) {};
		\node [style=ZXdimension] (4) at (-1, -0.475) {$a \plus 1$};
		\node [style=ZXdimension] (5) at (1, -0.475) {$a \plus 1$};
		\node [style=gn] (6) at (0, 0.575) {};
	\end{pgfonlayer}
	\begin{pgfonlayer}{edgelayer}
		\draw [bend left=90, looseness=1.50] (0.center) to (1.center);
		\draw (2.center) to (0.center);
		\draw (3.center) to (1.center);
	\end{pgfonlayer}
\end{tikzpicture}
 & \quad\quad\quad
  \input{figures/generators/ZW/swap-annot.tikz} &\quad\ZWtoZX\quad \begin{tikzpicture}
	\begin{pgfonlayer}{nodelayer}
		\node [style=none] (0) at (0.5, -0.75) {};
		\node [style=none] (3) at (-0.5, 0.75) {};
		\node [style=none] (4) at (-0.5, -0.75) {};
		\node [style=none] (5) at (0.5, 0.75) {};
		\node [style=ZXdimension] (6) at (-1, -0.75) {$b \plus 1$};
		\node [style=none] (7) at (-0.5, 1) {};
		\node [style=none] (8) at (0.5, 1) {};
		\node [style=none] (9) at (0.5, -1) {};
		\node [style=none] (10) at (-0.5, -1) {};
		\node [style=ZXdimension] (11) at (1, -0.75) {$a \plus 1$};
		\node [style=ZXdimension] (12) at (1, 0.75) {$b \plus 1$};
		\node [style=ZXdimension] (13) at (-1, 0.75) {$a \plus 1$};
	\end{pgfonlayer}
	\begin{pgfonlayer}{edgelayer}
		\draw [in=-90, out=90] (0.center) to (3.center);
		\draw [in=-90, out=90] (4.center) to (5.center);
		\draw (7.center) to (3.center);
		\draw (8.center) to (5.center);
		\draw (10.center) to (4.center);
		\draw (9.center) to (0.center);
	\end{pgfonlayer}
\end{tikzpicture}
 & \\
  \input{figures/generators/ZW/cup-annot.tikz} &\quad\ZWtoZX\quad \begin{tikzpicture}
	\begin{pgfonlayer}{nodelayer}
		\node [style=none] (0) at (-0.5, -0.125) {};
		\node [style=none] (1) at (0.5, -0.125) {};
		\node [style=none] (2) at (-0.5, 0.75) {};
		\node [style=none] (3) at (0.5, 0.75) {};
		\node [style=ZXdimension] (4) at (-1, 0.475) {$a \plus 1$};
		\node [style=ZXdimension] (5) at (1, 0.475) {$a \plus 1$};
		\node [style=gn] (6) at (0, -0.575) {};
	\end{pgfonlayer}
	\begin{pgfonlayer}{edgelayer}
		\draw [bend right=90, looseness=1.50] (0.center) to (1.center);
		\draw (2.center) to (0.center);
		\draw (3.center) to (1.center);
	\end{pgfonlayer}
\end{tikzpicture}
 & \quad\quad\quad
  \input{figures/generators/ZW/id-annot.tikz} &\quad\ZWtoZX\quad \begin{tikzpicture}
	\begin{pgfonlayer}{nodelayer}
		\node [style=none] (0) at (0, 0.75) {};
		\node [style=none] (1) at (0, -0.75) {};
		\node [style=ZXdimension] (2) at (-0.5, -0.5) {$a \plus 1$};
	\end{pgfonlayer}
	\begin{pgfonlayer}{edgelayer}
		\draw (0.center) to (1.center);
	\end{pgfonlayer}
\end{tikzpicture}
 & \quad\quad\quad
  r &\quad\ZWtoZX\quad \begin{tikzpicture}
	\begin{pgfonlayer}{nodelayer}
		\node [style=gbox] (0) at (0, -0.75) {$r$};
		\node [style={rn_phase}] (1) at (0, 0.75) {$K_1$};
		\node [style=ZXdimension] (2) at (-0.25, 0) {$2$};
	\end{pgfonlayer}
	\begin{pgfonlayer}{edgelayer}
		\draw (1) to (0);
	\end{pgfonlayer}
\end{tikzpicture}
&
\end{align*}

\begin{lemma}
  \label{lem:WX-sound}
  The ZW-to-ZX translation functor  preserves the semantics, that is, $\interp{\iWX[D]} = \interp{D}$ for any ZW-diagram $D$.
\end{lemma}
\noindent To verify this lemma, it suffices to apply the computational basis states to both the original ZW-diagram $D$ and its translated ZX-diagram $\iWX[D]$.
By evaluating these diagrams under all possible computational basis inputs, we can confirm that their interpretations yield identical results.

\subsection{Proof of Completeness}

In this section, we prove the completeness of $\ZXf$, the finite-dimensional ZX-calculus.
\begin{restatable}{lemma}{bnf}
  \label{lem:back-and-forth}
  For an arbitrary ZX-diagram $D$, $\ZXf \vdash \iXW[\iWX[D]] = D$.
\end{restatable}
\noindent Due to the functoriality of the monoidal functors $\iXW$ and $\iWX$, it suffices to show that the above lemma holds for all generators of $\ZXf$.
\cref{subsec:recovering-generators} discusses these lemmas.

\begin{restatable}{proposition}{rules}
  \label{prop:rule-interp}
  If $\ZWf \vdash D_1 = D_2$ then $\ZXf \vdash \iWX[D_1] = \iWX[D_2]$.
\end{restatable}
\noindent By the functoriality of $\iWX$, we only need to show that the translations of all axioms of $\ZWf$ are derivable in $\ZXf$.
These proofs are the content of \cref{subsec:proving-axioms}.

\begin{theorem}[Completeness]
  For finite-dimensional Hilbert spaces, the ZX-calculus is universally complete:
  For any two ZX-diagrams of the same type $D_1: A \to B$ and $D_2: A \to B$, if $\interp{D_1} = \interp{D_2}$, then $\ZXf \vdash D_1 = D_2$.
\end{theorem}
\begin{proof}
  Suppose $D_1, D_2 \in \ZXf$ such that $\interp{D_1} = \interp{D_2}$ and they have the same type.
  By \cref{lem:ZW-sound}, $\interp{\iXW[D_1]} = \interp{D_1} = \interp{D_2} = \interp{\iXW[D_2]}$.
  By the completeness of ZW-calculus, $\ZWf \vdash \iXW[D_1] = \iXW[D_2]$.
  Now by \cref{prop:rule-interp}, $\ZXf \vdash \iXW[\iWX[D_1]] = \iXW[\iWX[D_2]]$.
  Finally, by \cref{lem:back-and-forth}, $\ZXf \vdash \iXW[\iWX[D_1]] = D_1, ~ \vdash \iXW[\iWX[D_2]] = D_2$, therefore, $\ZXf \vdash D_1 = D_2$.
\end{proof}


\section{Conclusion and Further Work}\label{sec:future}

In this paper, we presented the finite-dimensional ZX-calculus, which generalizes the qudit ZX-calculus with mixed-dimensional Z-spiders.
We proved the completeness of the language by translating to and from the finite-dimensional ZW-calculus and showing the translation is invertible.

While the ZW-calculus is close to minimal, we did not consider the minimality of the ZX-calculus.
It would be interesting to investigate the smallest set of rules needed for the ZX-calculus to be complete for finite-dimensional Hilbert spaces.

Secondly, our calculus generalizes the Z-spider to mixed dimensions but leaves the X-spider unchanged from qudits.
Figuring out a nice mixed-dimensional generalization of the X-spider would allow us more flexibility in doing mixed-dimensional reasoning.

Next, we used the ZW-calculus to prove the completeness of the ZX-calculus.
By translating the ZW normal form, we easily obtain a normal form for the ZX-calculus.
This normal form can be used to prove the completeness of the ZX-calculus directly, by showing that every diagram can be reduced to the normal form.
We leave proving this direct completeness without relying on the result of a different calculus for future work.

Another idea to explore could be obtaining completeness through translation to and from the qufinite ZXW-calculus.
As it is a superset of finite-dimensional ZX-calculus, one side of the translation is trivial.
However, we suspect that given the close-minimality of the finite-dimensional ZW-calculus, the translation presented in this paper results in a more minimal set of rules.

Finally, we note that the ZX-calculus presented here allows for Z-boxes labelled with arbitrary complex numbers, in contrast with the original ZX-calculus where spider coefficients are restricted to be phases.
Further work could be done to probe whether the ZX-calculus with only phase coefficients is complete for finite-dimensional Hilbert spaces.

\section*{Acknowledgements}

We would like to thank Alexander Cowtan and Lia Yeh for their detailed feedback and several suggestions for improvement.
We thank the anonymous reviewers at QPL 2024 for their valuable feedback and pointing out a mistake in a translation of the previous version of this paper.
We would also like to thank the reviewers at QPL 2025 for their very detailed reviews and helpful suggestions.
BP is supported by the Engineering and Physical Sciences Research Council grant number EP/Z002230/1, \enquote{(De)constructing quantum software (DeQS)}.
RS is supported by the Clarendon Fund Scholarship.

\bibliographystyle{eptcs}
\bibliography{preamble/references-bibtex}

\newpage
\appendix

\allowdisplaybreaks
\setlength{\jot}{20pt}

\section{Lemmas}\label{sec:lemmas}

\subsection{General Derivations}

\begin{lemma}
  \label{scalargeneralmult}
  \input{figures/lemmas/scalargeneralmult.tikz}
\end{lemma}
\begin{proof}
  \[
    \input{figures/lemmas/scalargeneralmultprf.tikz}
  \]
\end{proof}


\begin{lemma}
  \label{scalarinverseditlm}
  Suppose $x \neq 0$.
  Then
  \input{figures/lemmas/scalarinversedit2.tikz}
\end{lemma}
\begin{proof}
  \[
    \input{figures/lemmas/scalarinverseditprf.tikz}
  \]
\end{proof}

\begin{lemma}
  \label{k1zbox1}
  \input{figures/lemmas/k1-zbox1.tikz}
\end{lemma}
\begin{proof}
  \[
    \input{figures/lemmas/k1-zbox1-pf.tikz}
  \]
\end{proof}

\begin{lemma}
  \label{hilm}\cite{wangQufiniteZXcalculusUnified2022}
  \input{figures/lemmas/h1scalar.tikz}
\end{lemma}
\begin{proof}
  \[
    \input{figures/lemmas/h1scalarprf.tikz}
  \]
\end{proof}

\begin{lemma}
  \label{hheqhdhd}
  \input{figures/lemmas/hheqhdhd.tikz}
\end{lemma}
\begin{proof}
  Same as \zxw{Lemma 12}.
\end{proof}

\begin{lemma}
  \label{dboxsquarelm}\cite{wangQufiniteZXcalculusUnified2022}
  \input{figures/lemmas/dboxsquare.tikz}
\end{lemma}
\begin{proof}
  Same as \zxw{Lemma 16}.
\end{proof}

\begin{lemma}
  \label{s4lm}\cite{wangQufiniteZXcalculusUnified2022}
  \begin{gather}
    \input{figures/lemmas/redspider0pfusedit2.tikz}
    \tag{S4}\label{rule:S4}
  \end{gather}
\end{lemma}
\begin{proof}
  \[
    \input{figures/lemmas/redspider0pfusedit2prf2.tikz}
  \]
  The other equalities can be proved similarly using the same rules.
\end{proof}

\begin{lemma}
  \label{redspiderforgrlm}\cite{wangQufiniteZXcalculusUnified2022}
  \begin{gather}
    \input{figures/lemmas/redspiderforgr2.tikz}
    \tag{HX}\label{rule:HX}
  \end{gather}
  where $v_{m,n}= d^{\frac{-m-n+2}{2}}$, $j \in \{ 0,1,\cdots, d-1\}$.
  We also call this equality~\eqref{rule:HX}.
\end{lemma}
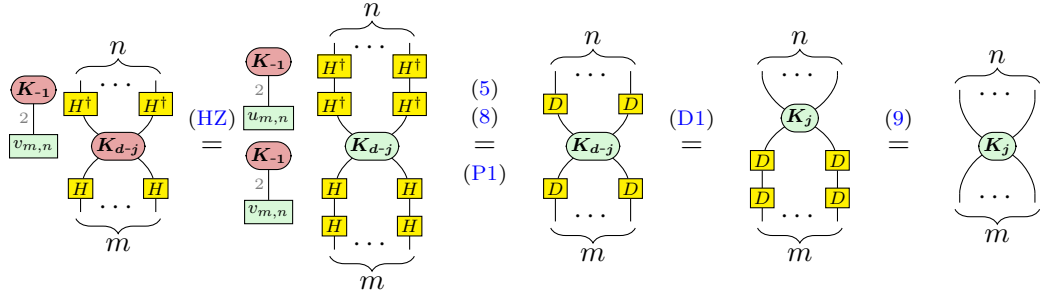
\begin{proof}
  \begin{align*}
    &\input{figures/lemmas/redspiderforgrprf-1v2.tikz} \\
    &\input{figures/lemmas/redspiderforgrprf-2v2.tikz}
  \end{align*}
  where $u_{m,n}= d^{\frac{m+n-2}{2}}$, $v_{m,n}= d^{\frac{-m-n+2}{2}}$, $j \in \{ 0,1,\cdots, d-1\}$.
\end{proof}

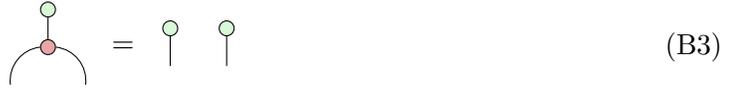
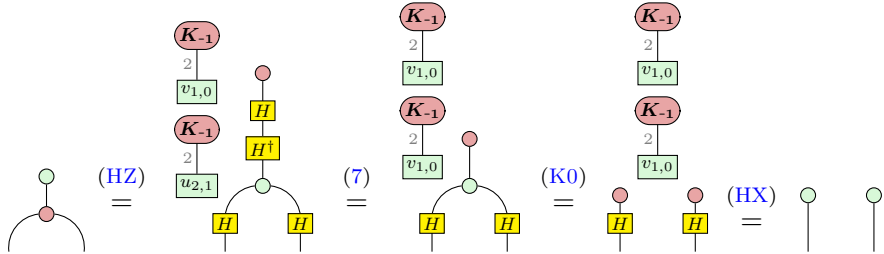
\begin{lemma}
  \label{b3quditlm}
  \begin{gather}
    \input{figures/lemmas/b3qudit.tikz}
    \tag{B3}\label{rule:B3}
  \end{gather}
\end{lemma}
\begin{proof}
  \[
    \input{figures/lemmas/b3quditprf2.tikz}
  \]
  where $u_{2,1}= d^{\frac{1}{2}}$, $v_{1,0}= d^{\frac{1}{2}}$.
\end{proof}

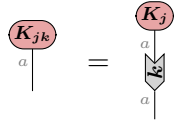
\begin{lemma}
  \label{multiplier-kj}
  \[
    \begin{tikzpicture}
	\begin{pgfonlayer}{nodelayer}
		\node [style=none] (0) at (0, 0) {$=$\strut};
		\node [style={rn_phase}] (1) at (-1.75, 0.75) {$K_{j k}$};
		\node [style=none] (2) at (-1.75, -0.75) {};
		\node [style=ZXdimension] (3) at (-2, 0) {$a$};
		\node [style=none] (14) at (1.5, -1.5) {};
		\node [style=ZXdimension] (15) at (1.25, -1) {$a$};
		\node [style={rn_phase}] (16) at (1.5, 1.25) {$K_{j}$};
		\node [style=dmat] (26) at (1.5, -0.25) {$k$};
		\node [style=ZXdimension] (27) at (1.25, 0.5) {$a$};
	\end{pgfonlayer}
	\begin{pgfonlayer}{edgelayer}
		\draw (1) to (2.center);
		\draw (26) to (16);
		\draw (26) to (14.center);
	\end{pgfonlayer}
\end{tikzpicture}

  \]
\end{lemma}
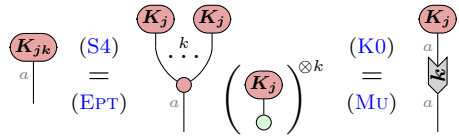
\begin{proof}
  \[
    \begin{tikzpicture}
	\begin{pgfonlayer}{nodelayer}
		\node [style=none] (0) at (0, 0) {$=$\strut};
		\node [style={rn_phase}] (1) at (-1.75, 0.75) {$K_{j k}$};
		\node [style=none] (2) at (-1.75, -0.75) {};
		\node [style=ZXdimension] (3) at (-2, 0) {$a$};
		\node [style=none] (4) at (2.25, -1.5) {};
		\node [style=ZXdimension] (5) at (2, -0.75) {$a$};
		\node [style=none] (6) at (1.5, 0.75) {};
		\node [style=rn] (7) at (2.25, -0.25) {};
		\node [style=none] (8) at (3, 0.75) {};
		\node [style={rn_phase}] (9) at (1.5, 1.5) {$K_{j}$};
		\node [style={rn_phase}] (10) at (3, 1.5) {$K_{j}$};
		\node [style=none] (11) at (2.25, 0.5) {$\ldots$};
		\node [style=wire label] (12) at (2.25, 0.875) {$k$};
		\node [style=none] (13) at (8, 0) {$=$\strut};
		\node [style=none] (14) at (9.75, -1.5) {};
		\node [style=ZXdimension] (15) at (9.5, -0.75) {$a$};
		\node [style={rn_phase}] (16) at (9.75, 1.5) {$K_{j}$};
		\node [style={rn_phase}] (17) at (4.375, 0) {$K_{j}$};
		\node [style=gn] (18) at (4.375, -1) {};
		\node [style=none] (19) at (3.5, -0.375) {\bigg(};
		\node [style=none] (20) at (5.25, -0.375) {\bigg)};
		\node [style=wire label] (21) at (6.125, 0.5) {$\otimes (k \minu 1)$};
		\node [style=wire label] (22) at (0, 0.75) {\axref{S4}};
		\node [style=wire label] (23) at (8, 0.75) {\axref{K0}};
		\node [style=wire label] (24) at (0, -0.75) {\axref{Ept}};
		\node [style=wire label] (25) at (8, -0.75) {\axref{Mu}};
		\node [style=dmat] (26) at (9.75, 0) {$k$};
		\node [style=ZXdimension] (27) at (9.5, 0.75) {$a$};
	\end{pgfonlayer}
	\begin{pgfonlayer}{edgelayer}
		\draw (1) to (2.center);
		\draw (7) to (4.center);
		\draw [in=150, out=-90, looseness=0.75] (6.center) to (7);
		\draw [in=30, out=-90, looseness=0.75] (8.center) to (7);
		\draw (9) to (6.center);
		\draw (10) to (8.center);
		\draw (17) to (18);
		\draw (26) to (16);
		\draw (26) to (14.center);
	\end{pgfonlayer}
\end{tikzpicture}

  \]
\end{proof}

\begin{lemma}
  \label{kjga1}\cite{wangQufiniteZXcalculusUnified2022}
  \input{figures/lemmas/kjga1.tikz}
  where $\overrightarrow{r}=(r_1,\cdots, r_{a-1})$.
\end{lemma}
\begin{proof}
  \[
    \input{figures/lemmas/kjga1prf.tikz}
  \]
  where $\overrightarrow{r_J} = (r_1, \cdots, r_j)$ and $\overleftarrow{r_J} = (r_j, \cdots, r_1)$.
\end{proof}


%
%
%
%
%
%
%
%
%
%


\begin{lemma}
  \label{xdualiserlm}
  \begin{align*}
    &\input{figures/lemmas/xdualiser-1.tikz}\\
    &\input{figures/lemmas/xdualiser-2.tikz}
  \end{align*}
\end{lemma}
\begin{proof}
  Same as \zxw{Lemma 25}.
\end{proof}

%
%
%
%

\begin{lemma}
  \label{hopfditlm}\cite{wangQufiniteZXcalculusUnified2022}
  \begin{gather}
    \input{figures/lemmas/hopfdit.tikz}
    \tag{Hopf}\label{rule:Hopf}
  \end{gather}
\end{lemma}
\begin{proof}
  Same as \zxw{Lemma 29}.
\end{proof}

\begin{lemma}
  \label{redgreenmchangelm}\cite{wangQufiniteZXcalculusUnified2022}
  \[
    \input{figures/lemmas/redgreenmchange.tikz}
  \]
  where $1 \leq k \leq d-1$.
\end{lemma}
\begin{proof}
  Same as \zxw{Lemma 30}.
\end{proof}

%
%
%
%

\begin{lemma}
  \label{dualisermultipliermultlm}
  Suppose $x \in \{0, \dotsc, d - 1 \}$.
  Then
  \[
    \begin{tikzpicture}
	\begin{pgfonlayer}{nodelayer}
		\node [style=dmat] (1) at (-3, 0.5) {$x$};
		\node [style=none] (3) at (-1.5, 0) {$=$};
		\node [style=none] (5) at (0, -1.5) {};
		\node [style=hbox] (6) at (-3, -0.75) {$D$};
		\node [style=none] (7) at (-3, -1.5) {};
		\node [style=none] (8) at (-3, 1.5) {};
		\node [style=dmat] (9) at (0, 0) {$\minu x$};
		\node [style=none] (10) at (0, 1.5) {};
		\node [style=dmat] (11) at (3, -0.5) {$x$};
		\node [style=none] (12) at (1.5, 0) {$=$};
		\node [style=hbox] (13) at (3, 0.75) {$D$};
		\node [style=none] (14) at (3, 1.5) {};
		\node [style=none] (15) at (3, -1.5) {};
	\end{pgfonlayer}
	\begin{pgfonlayer}{edgelayer}
		\draw (1) to (6);
		\draw (8.center) to (1);
		\draw (7.center) to (6);
		\draw (10.center) to (9);
		\draw (5.center) to (9);
		\draw (11) to (13);
		\draw (15.center) to (11);
		\draw (14.center) to (13);
	\end{pgfonlayer}
\end{tikzpicture}

  \]
\end{lemma}
\begin{proof}
  Same as \zxw{Lemma 34}.
\end{proof}

\begin{lemma}
  \label{multiplierkjcap}
  \[
    \begin{tikzpicture}
	\begin{pgfonlayer}{nodelayer}
		\node [style={rn_phase}] (0) at (-2.5, 1.25) {$K_{\minu j k}$};
		\node [style=umat] (1) at (-1.5, -0.25) {$j$};
		\node [style=none] (2) at (-1.5, 0.25) {};
		\node [style=none] (3) at (-3.5, 0.25) {};
		\node [style=none] (4) at (-1.5, -1.25) {};
		\node [style=none] (5) at (-3.5, -1.25) {};
		\node [style=none] (6) at (0, 0) {$=$\strut};
		\node [style={rn_phase}] (7) at (2.5, 1.25) {$K_{\minu k}$};
		\node [style=none] (9) at (3.5, 0.25) {};
		\node [style=none] (10) at (1.5, 0.25) {};
		\node [style=none] (11) at (3.5, -1.25) {};
		\node [style=none] (12) at (1.5, -1.25) {};
		\node [style=dmat] (13) at (1.5, -0.25) {$j$};
	\end{pgfonlayer}
	\begin{pgfonlayer}{edgelayer}
		\draw (2.center) to (1);
		\draw (4.center) to (1);
		\draw [bend right=90, looseness=1.75] (2.center) to (3.center);
		\draw (5.center) to (3.center);
		\draw [bend right=90, looseness=1.75] (9.center) to (10.center);
		\draw (12.center) to (10.center);
		\draw (9.center) to (11.center);
	\end{pgfonlayer}
\end{tikzpicture}

  \]
\end{lemma}
\begin{proof}
  \[
    \begin{tikzpicture}
	\begin{pgfonlayer}{nodelayer}
		\node [style={rn_phase}] (0) at (-2.5, 1.25) {$K_{\minu j k}$};
		\node [style=umat] (1) at (-1.5, -0.25) {$j$};
		\node [style=none] (2) at (-1.5, 0.25) {};
		\node [style=none] (3) at (-3.5, 0.25) {};
		\node [style=none] (4) at (-1.5, -1.5) {};
		\node [style=none] (5) at (-3.5, -1.5) {};
		\node [style=none] (6) at (0, 0) {$=$\strut};
		\node [style=wire label] (7) at (0, -0.75) {\axref{S4}};
		\node [style=wire label] (8) at (0, 0.75) {\axref{Mu}};
		\node [style=rn] (9) at (2.5, 2) {};
		\node [style=none] (11) at (3.5, 1) {};
		\node [style=none] (12) at (1.5, 1) {};
		\node [style=none] (13) at (3.5, -2.25) {};
		\node [style=none] (14) at (1.5, -2.25) {};
		\node [style=rn] (17) at (3.5, 0.875) {};
		\node [style=wire label] (18) at (3.5, -1) {$\ldots$};
		\node [style=wire label] (19) at (3.5, -0.625) {$j$};
		\node [style=gn] (20) at (3.5, -1.625) {};
		\node [style={rn_phase}] (21) at (2.75, -0.125) {$K_{\minu k}$};
		\node [style={rn_phase}] (22) at (4.25, -0.125) {$K_{\minu k}$};
		\node [style=none] (23) at (5.75, 0) {$=$\strut};
		\node [style=wire label] (24) at (5.75, 0.75) {\axref{K2}};
		\node [style=rn] (26) at (8, 1.75) {};
		\node [style=none] (27) at (8.75, 1) {};
		\node [style=none] (28) at (7.25, 1) {};
		\node [style=none] (30) at (7.25, -2.25) {};
		\node [style=none] (37) at (8.75, -2.25) {};
		\node [style=gn] (38) at (8.75, -0.5) {};
		\node [style=rn] (39) at (8.75, 1) {};
		\node [style=wire label] (40) at (8.75, 0.1) {$\ldots$};
		\node [style=wire label] (41) at (8.75, 0.375) {$j$};
		\node [style={rn_phase}] (45) at (8.75, -1.375) {$K_{\minu k}$};
		\node [style=none] (46) at (10.25, 0) {$=$\strut};
		\node [style=wire label] (47) at (10.25, -0.75) {\lemref{redgreenmchangelm}};
		\node [style=none] (50) at (13.25, 1) {};
		\node [style=none] (51) at (11.75, 1) {};
		\node [style=none] (52) at (11.75, -2.25) {};
		\node [style=none] (53) at (13.25, -2.25) {};
		\node [style=gn] (54) at (11.75, 0.75) {};
		\node [style=rn] (55) at (11.75, -1) {};
		\node [style=wire label] (56) at (11.75, -0.275) {$\ldots$};
		\node [style=wire label] (57) at (11.75, 0) {$d \minu j$};
		\node [style={rn_phase}] (58) at (13.25, 0.375) {$K_{\minu k}$};
		\node [style=wire label] (59) at (10.25, 0.75) {\axref{S4}};
		\node [style=none] (60) at (14.75, 0) {$=$\strut};
		\node [style=none] (62) at (17.75, 1) {};
		\node [style=none] (63) at (16.25, 1) {};
		\node [style=none] (64) at (16.25, -2.25) {};
		\node [style=none] (65) at (17.75, -2.25) {};
		\node [style={rn_phase}] (70) at (17, 1.875) {$K_{\minu k}$};
		\node [style=hbox] (71) at (16.25, 0.5) {$D$};
		\node [style=wire label] (72) at (14.75, 0.75) {\lemref{xdualiserlm}};
		\node [style=wire label] (73) at (19.25, 0.75) {\lemref{dualisermultipliermultlm}};
		\node [style=dmat] (74) at (16.25, -1) {$\minu j$};
		\node [style=wire label] (75) at (14.75, -0.75) {\axref{Mu}};
		\node [style=none] (76) at (19.25, 0) {$=$\strut};
		\node [style=none] (77) at (22.25, 0.5) {};
		\node [style=none] (78) at (20.75, 0.5) {};
		\node [style=none] (79) at (20.75, -1.75) {};
		\node [style=none] (80) at (22.25, -1.75) {};
		\node [style={rn_phase}] (81) at (21.5, 1.375) {$K_{\minu k}$};
		\node [style=dmat] (84) at (20.75, -0.5) {$j$};
	\end{pgfonlayer}
	\begin{pgfonlayer}{edgelayer}
		\draw (2.center) to (1);
		\draw (4.center) to (1);
		\draw [bend right=90, looseness=1.75] (2.center) to (3.center);
		\draw (5.center) to (3.center);
		\draw [bend right=90, looseness=1.75] (11.center) to (12.center);
		\draw (14.center) to (12.center);
		\draw [in=90, out=-150, looseness=0.75] (17) to (21);
		\draw [in=90, out=-30, looseness=0.75] (17) to (22);
		\draw [in=45, out=-90] (22) to (20);
		\draw [in=-90, out=135] (20) to (21);
		\draw (11.center) to (17);
		\draw (20) to (13.center);
		\draw [bend right=90, looseness=1.75] (27.center) to (28.center);
		\draw (30.center) to (28.center);
		\draw [bend right=315] (39) to (38);
		\draw [bend left=45] (38) to (39);
		\draw (37.center) to (38);
		\draw [bend right=90, looseness=1.75] (50.center) to (51.center);
		\draw [bend left=315] (55) to (54);
		\draw [bend right=45] (54) to (55);
		\draw (55) to (52.center);
		\draw (51.center) to (54);
		\draw (53.center) to (58);
		\draw (50.center) to (58);
		\draw [bend right=90, looseness=1.75] (62.center) to (63.center);
		\draw (62.center) to (65.center);
		\draw (64.center) to (74);
		\draw (74) to (71);
		\draw (71) to (63.center);
		\draw [bend right=90, looseness=1.75] (77.center) to (78.center);
		\draw (77.center) to (80.center);
		\draw (79.center) to (84);
		\draw (78.center) to (84);
	\end{pgfonlayer}
\end{tikzpicture}

  \]
\end{proof}

%
%

\subsection{Recovering Generators}\label{subsec:recovering-generators}

\begin{lemma}
  \label{lem:Z-state}
  \[
    \iXW[\iWX[\begin{tikzpicture}
	\begin{pgfonlayer}{nodelayer}
		\node [style=gbox] (0) at (0, 0.5) {$\overrightarrow{r}$};
		\node [style=none] (1) at (0, -0.75) {};
		\node [style=ZXdimension] (2) at (-0.5, -0.375) {$a \plus 1$};
	\end{pgfonlayer}
	\begin{pgfonlayer}{edgelayer}
		\draw (1.center) to (0);
	\end{pgfonlayer}
\end{tikzpicture}
\ ]]
    \quad=\quad
    \begin{tikzpicture}
	\begin{pgfonlayer}{nodelayer}
		\node [style=gbox] (0) at (0, 0.5) {$\overrightarrow{r}$};
		\node [style=none] (1) at (0, -0.75) {};
		\node [style=ZXdimension] (2) at (-0.5, -0.375) {$a \plus 1$};
	\end{pgfonlayer}
	\begin{pgfonlayer}{edgelayer}
		\draw (1.center) to (0);
	\end{pgfonlayer}
\end{tikzpicture}

  \]
\end{lemma}
\begin{proof}
  \begin{align*}
    \begin{tikzpicture}
	\begin{pgfonlayer}{nodelayer}
		\node [style=gbox] (0) at (0, 0.5) {$\overrightarrow{r}$};
		\node [style=none] (1) at (0, -0.75) {};
		\node [style=ZXdimension] (2) at (-0.5, -0.375) {$a \plus 1$};
	\end{pgfonlayer}
	\begin{pgfonlayer}{edgelayer}
		\draw (1.center) to (0);
	\end{pgfonlayer}
\end{tikzpicture}

    \quad &\ZXtoZW \quad
    \input{figures/translation/from-ZW/Z-box-state.tikz}
    \quad\ZWtoZX\quad
    \input{figures/proofs/Generators/Z-box-state-0.tikz}\\
    \quad &
    \input{figures/proofs/Generators/Z-box-state.tikz}
  \end{align*}
\end{proof}

\begin{lemma}
  \label{lem:Z-embed}
  \[
    \iXW[\iWX[\input{figures/generators/ZX/Z-embed.tikz}\ ]]
    \quad=\quad
    \input{figures/generators/ZX/Z-embed.tikz}
  \]
\end{lemma}
\begin{proof}
  Suppose $a \geq b$. Then
  \begin{align*}
    \input{figures/generators/ZX/Z-embed.tikz}
    \quad\ZXtoZW\quad
    \input{figures/translation/from-ZW/Z-embed-1.tikz}
    \quad\ZWtoZX\quad
    \input{figures/proofs/Generators/Z-embed.tikz}
  \end{align*}
  The other case can be proved similarly.
\end{proof}

\begin{lemma}
  \label{lem:Z-mult}
  \[
    \iXW[\iWX[\input{figures/generators/ZX/Z-mult.tikz}\ ]]
    \quad=\quad
    \input{figures/generators/ZX/Z-mult.tikz}
  \]
\end{lemma}
\begin{proof}
  For $C = \left(\frac{1}{\sqrt{N!}}\right)^{m+n-1}$,
  \begin{align*}
    \input{figures/generators/ZX/Z-mult.tikz}
    \quad\ZXtoZW\quad
    \input{figures/translation/from-ZW/phase-Z-box.tikz}
    \ZWtoZX
    \input{figures/proofs/Generators/Z-box-0.tikz}\\
  \end{align*}
\end{proof}

\begin{lemma}
  \label{lem:Z-box}
  \[
    \iXW[\iWX[\input{figures/generators/ZX/mixed-Z-box.tikz}\ ]]
    \quad=\quad
    \input{figures/generators/ZX/mixed-Z-box.tikz}
  \]
\end{lemma}
\begin{proof}
  This follows from \cref{lem:Z-state,lem:Z-embed,lem:Z-mult}.
\end{proof}

\begin{lemma}
  \label{lem:mod}
  \[
    \normalfont
    \iWX[\input{figures/generators/ZW/mod-d.tikz}\ ]
    \quad=\quad
    \input{figures/definitions/ZX/mod-d-expanded.tikz}
  \]
\end{lemma}
\begin{proof}
  Expanding the definition of the modulo $a$ box,
  \begin{align*}
    \begin{tikzpicture}
	\begin{pgfonlayer}{nodelayer}
		\node [style=W-1-n] (0) at (0, 0.5) {};
		\node [style=none] (1) at (-0.25, -0.375) {\scalebox{.7}{$\ldots$}};
		\node [style=none] (2) at (0.75, -0.5) {};
		\node [style=none] (3) at (0.75, -3.375) {};
		\node [style=none] (4) at (0, 3.25) {};
		\node [style=ZXdimension] (8) at (-0.25, -0.125) {$a$};
		\node [style=rZWdimension] (9) at (-0.5, -1) {$1$~~};
		\node [style=white dot] (10) at (-0.5, -1) {};
		\node [style=gbox] (11) at (0, 2.25) {\tiny $\frac{1}{\sNf}$};
		\node [style=gbox] (12) at (0.75, -2.25) {$\sNf$};
		\node [style=none] (14) at (-1.1, 2.25) {$\Big\llbracket$};
		\node [style=none] (15) at (1.65, 2.25) {$\Big\rrbracket_{XW\!}$};
		\node [style=none] (16) at (-0.35, -2.25) {$\Big\llbracket$};
		\node [style=none] (17) at (2.4, -2.25) {$\Big\rrbracket_{XW\!}$};
		\node [style=rZWdimension] (18) at (0, 1) {$2a \minu 1$~};
		\node [style=lZWdimension] (19) at (0.75, -0.999) {$~a \minu 1$};
	\end{pgfonlayer}
	\begin{pgfonlayer}{edgelayer}
		\draw (3.center) to (2.center);
		\draw [bend left=45] (10) to (0);
		\draw [bend right=45] (10) to (0);
		\draw [in=90, out=-30, looseness=0.50] (0) to (2.center);
		\draw (4.center) to (0);
	\end{pgfonlayer}
\end{tikzpicture}

    \quad\ZWtoZX\quad
    \input{figures/proofs/Generators/mod-d.tikz}
  \end{align*}
\end{proof}

\begin{lemma}
  \[
    \iXW[\iWX[\input{figures/generators/ZX/X-spider-2-1.tikz}\ ]]
    \quad=\quad
    \input{figures/generators/ZX/X-spider-2-1.tikz}
  \]
\end{lemma}
\begin{proof}
  \begin{align*}
    \input{figures/generators/ZX/X-spider-2-1.tikz}
    \quad\ZXtoZW\quad
    \input{figures/translation/from-ZW/X-spider-2-1.tikz}
    \quad\ZWtoZX\quad
    \input{figures/proofs/Generators/X-spider-2-1.tikz}\\
    \input{figures/proofs/Generators/X-spider-2-1-2.tikz}
  \end{align*}
\end{proof}

\begin{lemma}
  \[
    \iXW[\iWX[\begin{tikzpicture}
	\begin{pgfonlayer}{nodelayer}
		\node [style=none] (0) at (0, 1) {};
		\node [style=none] (1) at (0, -1) {};
		\node [style=ZXdimension] (2) at (-0.25, -0.75) {$a$};
	\end{pgfonlayer}
	\begin{pgfonlayer}{edgelayer}
		\draw (0.center) to (1.center);
	\end{pgfonlayer}
\end{tikzpicture}
\ ]]
    \quad=\quad
    \begin{tikzpicture}
	\begin{pgfonlayer}{nodelayer}
		\node [style=none] (0) at (0, 1) {};
		\node [style=none] (1) at (0, -1) {};
		\node [style=ZXdimension] (2) at (-0.25, -0.75) {$a$};
	\end{pgfonlayer}
	\begin{pgfonlayer}{edgelayer}
		\draw (0.center) to (1.center);
	\end{pgfonlayer}
\end{tikzpicture}

  \]
\end{lemma}
\begin{proof}
  \begin{align*}
    \begin{tikzpicture}
	\begin{pgfonlayer}{nodelayer}
		\node [style=none] (0) at (0, 1) {};
		\node [style=none] (1) at (0, -1) {};
		\node [style=ZXdimension] (2) at (-0.25, -0.75) {$a$};
	\end{pgfonlayer}
	\begin{pgfonlayer}{edgelayer}
		\draw (0.center) to (1.center);
	\end{pgfonlayer}
\end{tikzpicture}

    \quad\ZXtoZW\quad
    \input{figures/translation/from-ZW/id.tikz}
    \quad\ZWtoZX\quad
    \begin{tikzpicture}
	\begin{pgfonlayer}{nodelayer}
		\node [style=none] (0) at (0, 1) {};
		\node [style=none] (1) at (0, -1) {};
		\node [style=ZXdimension] (2) at (-0.25, -0.75) {$a$};
	\end{pgfonlayer}
	\begin{pgfonlayer}{edgelayer}
		\draw (0.center) to (1.center);
	\end{pgfonlayer}
\end{tikzpicture}

  \end{align*}
\end{proof}

\begin{lemma}
  \label{lem:zx-braid}
  \[
    \iXW[\iWX[\input{figures/generators/ZX/braid.tikz}\ ]]
    \quad=\quad
    \input{figures/generators/ZX/braid.tikz}
  \]
\end{lemma}
\begin{proof}
  \begin{align*}
    \input{figures/generators/ZX/braid.tikz}
    \quad\ZXtoZW\quad
    \input{figures/translation/from-ZW/braid.tikz}
    \quad\ZWtoZX\quad
    \input{figures/generators/ZX/braid.tikz}
  \end{align*}
\end{proof}

\bnf*
\begin{proof}
  Due to the functoriality of the monoidal functors $\iWX$ and $\iXW$,
  it suffices to show that the above lemma holds for all generators of $\ZXf$.
  These generators are presented in \cref{subsec:zx-generators}, and the proofs are shown in~\crefrange{lem:Z-state}{lem:zx-braid}.
\end{proof}

\subsection{Proving Axioms of ZW}\label{subsec:proving-axioms}

\begin{lemma}
  \label{lem:Z-spider-rule}
  The translation of the following ZW diagrams under $\iWX$ equal in $\ZXf$:
  \[
    \input{figures/axioms/ZW/Z-spider-rule-annot-00.tikz}
    \qquad \qquad \qquad \qquad
    \input{figures/axioms/ZW/Z-spider-rule-annot-01.tikz}
  \]
\end{lemma}
\begin{proof}
  \begin{align*}
    \input{figures/axioms/ZW/Z-spider-rule-annot-00.tikz}
    \quad\ZWtoZX\quad
    \input{figures/proofs/ZWAxioms/fusion.tikz}
    \quad\ZXfromXW\quad
    \input{figures/axioms/ZW/Z-spider-rule-annot-01.tikz}
  \end{align*}
\end{proof}

\begin{lemma}
  The translation of the following ZW diagrams under $\iWX$ equal in $\ZXf$:
  \[
    \input{figures/axioms/ZW/Z-id-annot-02.tikz}
    \qquad \qquad \qquad \qquad
    \input{figures/axioms/ZW/Z-id-annot-01.tikz}
  \]
\end{lemma}
\begin{proof}
  \begin{align*}
    \input{figures/axioms/ZW/Z-id-annot-02.tikz}
    \quad\ZWtoZX\quad
    \input{figures/proofs/ZWAxioms/id.tikz}
    \quad\ZXfromXW\quad
    \input{figures/axioms/ZW/Z-id-annot-01.tikz}
  \end{align*}
\end{proof}

\begin{lemma}
  \label{lem:ZW-a}
  For $b \geq \min(c, \sum a_i)$, the translation of the following ZW diagrams under $\iWX$ equal in $\ZXf$:
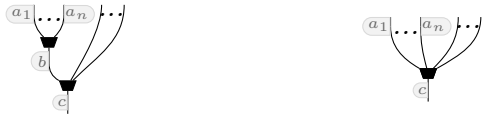
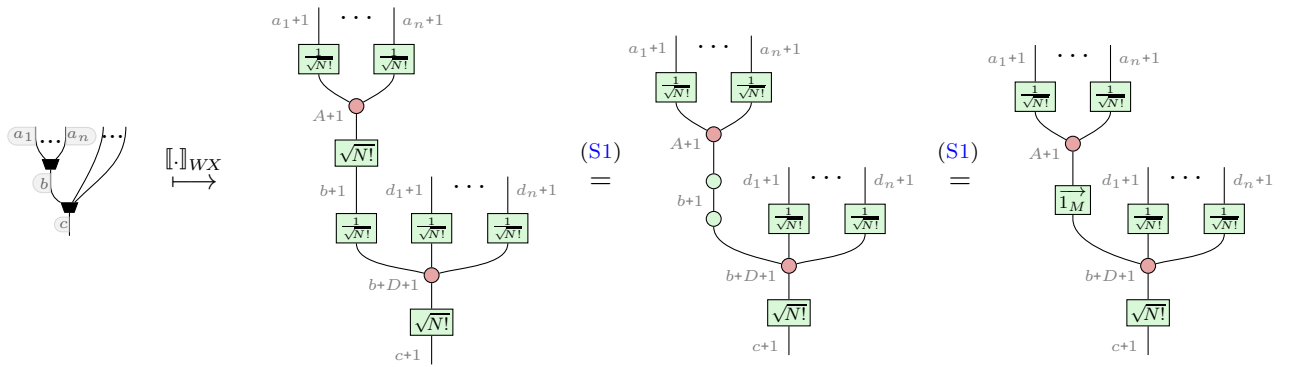
  \[
    \input{figures/axioms/ZW/W-assoc-annot-00.tikz}
    \qquad \qquad \qquad \qquad
    \input{figures/axioms/ZW/W-assoc-annot-01.tikz}
  \]
\end{lemma}
\begin{proof}
  For $A \coloneqq \sum a_i$, $D \coloneqq \sum d_i$, $M = \min\{A, b\}:$
  \begin{align*}
    \input{figures/axioms/ZW/W-assoc-annot-00.tikz}
    \quad\ZWtoZX\quad
    \input{figures/proofs/ZWAxioms/W-assoc.tikz}
  \end{align*}
  Here, we split the proof based on the value of $M$.
  If $A \leq b$, then $M = \min\{A, b\} = A$ and we have:
  \begin{align*}
    \input{figures/proofs/ZWAxioms/W-assoc-2.tikz}
  \end{align*}
  Otherwise, we have $A > b$, and thus $M = \min\{A, b\} = b$.
  Combining this with the assumption $b \geq \min(c, \sum a_i) = \min(c, A)$, we conclude that $b \geq c$.
  Therefore,
  \begin{align*}
    \scalebox{.9}{\input{figures/proofs/ZWAxioms/W-assoc-3.tikz}}
  \end{align*}
  Then, the remaining steps are the same as in the $A \leq b$ case.
\end{proof}

\begin{lemma}
  For $n \neq 0$, the translation of the following ZW diagrams under $\iWX$ equal in $\ZXf$:
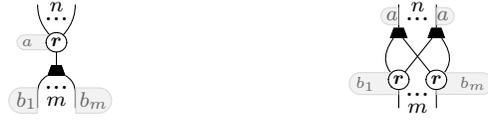
  \[
    \input{figures/axioms/ZW/Z-W-bialgebra-annot-00.tikz}
    \qquad \qquad \qquad \qquad
    \input{figures/axioms/ZW/Z-W-bialgebra-annot-01.tikz}
  \]
\end{lemma}
\begin{proof}
  Let $b \coloneqq \sum b_i$.
  We first show the following:
  \begin{align*}
    \input{figures/axioms/ZW/Z-W-bialgebra-phaseless-annot-00.tikz}
    \quad\ZWtoZX\quad
    \input{figures/proofs/ZWAxioms/Z-W-bialgebra-phaseless.tikz}\\
    \quad\ZXfromXW\quad
    \input{figures/axioms/ZW/Z-W-bialgebra-phaseless-annot-01.tikz}
  \end{align*}
  Now, we only have to show that the W-node copies phases:
  \begin{align*}
    \input{figures/axioms/ZW/phase-W-bialgebra-annot-00.tikz}
    \quad\ZWtoZX\quad
    \input{figures/proofs/ZWAxioms/phase-W-bialgebra.tikz}\\
    \quad\ZXfromXW\quad
    \input{figures/axioms/ZW/phase-W-bialgebra-annot-01.tikz}
  \end{align*}
\end{proof}

\begin{lemma}
  The translation of the following ZW diagrams under $\iWX$ equal in $\ZXf$:
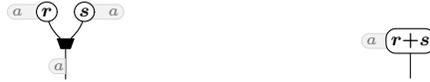
  \[
    \input{figures/axioms/ZW/sum-annot-00.tikz}
    \qquad \qquad \qquad \qquad
    \begin{tikzpicture}
	\begin{pgfonlayer}{nodelayer}
		\node [style=rZWdimension] (12)  at (-0.124, 0.124) {$a$~~.};
		\node [style=white phase dot] (5)  at (0.124, 0.124) {$r{+}s$};
		\node [style=none] (6)  at (0.124, -0.876) {};
		\node [style=none] (7)  at (0.124, 0.876) {};
	\end{pgfonlayer}
	\begin{pgfonlayer}{edgelayer}
		\draw (5) to (6.center);
	\end{pgfonlayer}
\end{tikzpicture}
  \]
\end{lemma}
\begin{proof}
  \begin{align*}
    \input{figures/axioms/ZW/sum-annot-00.tikz}
    \quad\ZWtoZX\quad
    \input{figures/proofs/ZWAxioms/sum.tikz}
    \quad\ZXfromXW\quad
    
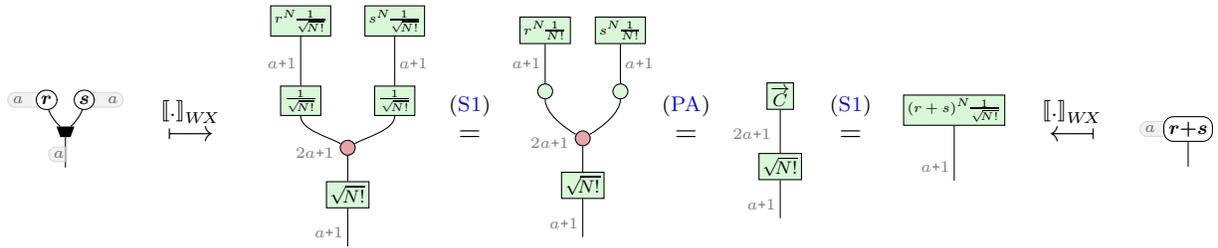
\begin{tikzpicture}
	\begin{pgfonlayer}{nodelayer}
		\node [style=rZWdimension] (12)  at (-0.124, 0.124) {$a$~~.};
		\node [style=white phase dot] (5)  at (0.124, 0.124) {$r{+}s$};
		\node [style=none] (6)  at (0.124, -0.876) {};
		\node [style=none] (7)  at (0.124, 0.876) {};
	\end{pgfonlayer}
	\begin{pgfonlayer}{edgelayer}
		\draw (5) to (6.center);
	\end{pgfonlayer}
\end{tikzpicture}
  \end{align*}
  where the $k$-th element of $\overrightarrow C$ for $k \leq a$ is given by $\displaystyle c_k = \sum_{i = 0}^k r^i \frac{1}{i!} s^{k \minu i} \frac{1}{(k - i)!} = \frac{(r + s)^k}{k!}$
  and for $k > a$ it is $\displaystyle c_k = \sum_{i = k - a}^a r^i \frac{1}{i!} s^{k \minu i} \frac{1}{(k - i)!}$; therefore, $\frac{(r + s)^N}{N!}$ defines the first $a + 1$ element of $\overrightarrow C$.
\end{proof}

\begin{lemma}
  The translation of the global scalar $1$ and the empty diagram $\input{figures/axioms/ZW/empty.tikz}$ under $\iWX$ equal in $\ZXf$.
\end{lemma}
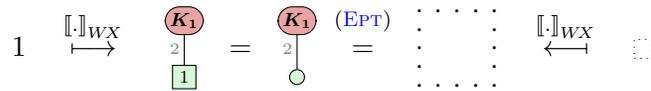
\begin{proof}
  \begin{align*}
    1
    \quad\ZWtoZX\quad
    \input{figures/proofs/ZWAxioms/empty.tikz}
    \quad\ZXfromXW\quad
    \input{figures/axioms/ZW/empty.tikz}
  \end{align*}
\end{proof}

\begin{lemma}
  The translation of the following ZW diagrams under $\iWX$ equal in $\ZXf$:
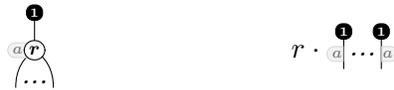
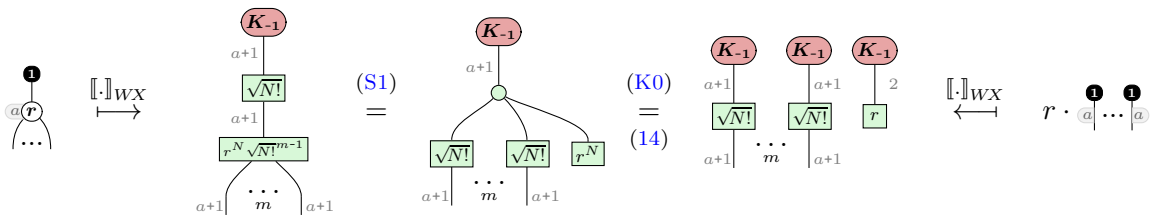
  \[
    \input{figures/axioms/ZW/copy-1-annot-00.tikz}
    \qquad \qquad \qquad \qquad
    r \cdot \input{figures/axioms/ZW/copy-1-annot-01.tikz}
  \]
\end{lemma}
\begin{proof}
  \begin{align*}
    \input{figures/axioms/ZW/copy-1-annot-00.tikz}
    \quad\ZWtoZX\quad
    \input{figures/proofs/ZWAxioms/copy-1.tikz}
    \quad\ZXfromXW\quad
    r \cdot \input{figures/axioms/ZW/copy-1-annot-01.tikz}
  \end{align*}
\end{proof}

\begin{lemma}
  For $c\geq\min(\sum a_i, \sum b_i)$, $\ell_{ij}=\min(a_i,b_j)$, the translation of the following ZW diagrams under $\iWX$ equal in $\ZXf$:
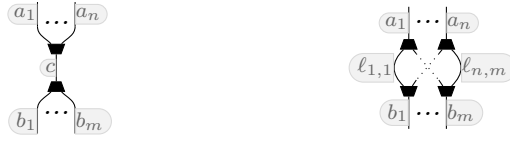
  \[
    \input{figures/axioms/ZW/W-bialgebra-annot-00.tikz}
    \qquad \qquad \qquad \qquad
    \input{figures/axioms/ZW/W-bialgebra-annot-01.tikz}
  \]
\end{lemma}
\begin{proof}
  This rule follows from Axiom~\eqref{rule:WW}.
\end{proof}

\begin{lemma}
  The translation of the following ZW diagrams under $\iWX$ equal in $\ZXf$:
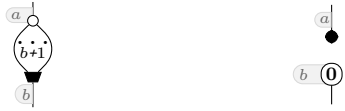
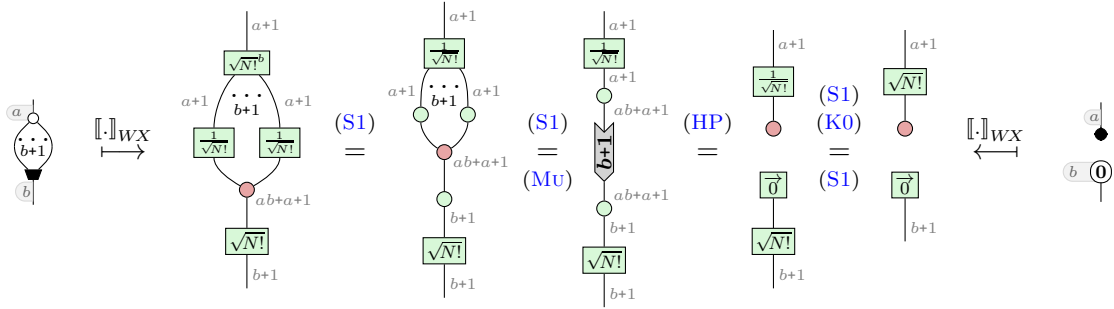
  \[
    \input{figures/axioms/ZW/Hopf-annot-00.tikz}
    \qquad \qquad \qquad \qquad
    \input{figures/axioms/ZW/Hopf-annot-01.tikz}
  \]
\end{lemma}
\begin{proof}
  \begin{align*}
    \input{figures/axioms/ZW/Hopf-annot-00.tikz}
    \quad\ZWtoZX
    \ \input{figures/proofs/ZWAxioms/Hopf.tikz}
    \ \ZXfromXW\quad
    \input{figures/axioms/ZW/Hopf-annot-01.tikz}
  \end{align*}
\end{proof}

\begin{lemma}
  For $0< k \leq a$, the translation of the following ZW diagrams under $\iWX$ equal in $\ZXf$:
  \[
    \input{figures/axioms/ZW/ket-1-ket-k-annot-00.tikz}
    \qquad \qquad \qquad \qquad
    \input{figures/axioms/ZW/ket-1-ket-k-annot-01.tikz}
  \]
\end{lemma}
\begin{proof}
  \begin{align*}
    \input{figures/axioms/ZW/ket-1-ket-k-annot-00.tikz}
    \quad\ZWtoZX\quad
    \input{figures/proofs/ZWAxioms/ket-1-ket-k-1.tikz}\\
    \input{figures/proofs/ZWAxioms/ket-1-ket-k-2.tikz}\\
    \input{figures/proofs/ZWAxioms/ket-1-ket-k-3.tikz}\\
    \hspace{-2cm}\quad\ZXfromXW\quad
    \input{figures/axioms/ZW/ket-1-ket-k-annot-01.tikz}
  \end{align*}
\end{proof}

\begin{lemma}
  The translation of the following ZW diagrams under $\iWX$ equal in $\ZXf$:
  \[
    \begin{tikzpicture}
	\begin{pgfonlayer}{nodelayer}
		\node [style=rZWdimension] (14)  at (0.0, 0.3) {$1~$};
		\node [style=white dot] (6)  at (0.0, 0.5) {};
		\node [style=none] (7)  at (0.0, -0.5) {};
	\end{pgfonlayer}
	\begin{pgfonlayer}{edgelayer}
		\draw (7.center) to (6);
	\end{pgfonlayer}
\end{tikzpicture}
    \qquad \qquad \qquad \qquad
    \input{figures/axioms/ZW/e1-NF-annot-01.tikz}
  \]
\end{lemma}
\begin{proof}
  \begin{align*}
    \begin{tikzpicture}
	\begin{pgfonlayer}{nodelayer}
		\node [style=rZWdimension] (14)  at (0.0, 0.3) {$1~$};
		\node [style=white dot] (6)  at (0.0, 0.5) {};
		\node [style=none] (7)  at (0.0, -0.5) {};
	\end{pgfonlayer}
	\begin{pgfonlayer}{edgelayer}
		\draw (7.center) to (6);
	\end{pgfonlayer}
\end{tikzpicture}
    \quad\ZWtoZX\quad
    \input{figures/proofs/ZWAxioms/e1-NF.tikz}
    \ \ZXfromXW
    \ \input{figures/axioms/ZW/e1-NF-annot-01.tikz}
  \end{align*}
\end{proof}

\begin{lemma}
  The translation of the following ZW diagrams under $\iWX$ equal in $\ZXf$:
  \[
    \input{figures/axioms/ZW/W-assoc-annot2-00.tikz}
    \qquad \qquad \qquad \qquad
    \input{figures/axioms/ZW/W-assoc-annot2-01.tikz}
  \]
\end{lemma}
\begin{proof}
  Let $a \coloneqq a_0 + a_1$, then
  \begin{align*}
    \input{figures/axioms/ZW/W-assoc-annot2-00.tikz}
    \quad\ZWtoZX\quad
    &\input{figures/proofs/ZWAxioms/W-assoc2.tikz} \\
    &\input{figures/proofs/ZWAxioms/W-assoc2-2.tikz}
    \quad\ZXfromXW\quad
    \input{figures/axioms/ZW/W-assoc-annot2-01.tikz}
  \end{align*}
\end{proof}

\begin{lemma}
  The translation of the following ZW diagrams under $\iWX$ equal in $\ZXf$:
  \[
    \input{figures/axioms/ZW/Z-copy-annot-00.tikz}
    \qquad \qquad \qquad \qquad
    \input{figures/axioms/ZW/Z-copy-annot-01.tikz}
  \]
\end{lemma}
\begin{proof}
  Let $a = a_0 + a_1$, then
  \begin{align*}
    \input{figures/axioms/ZW/Z-copy-annot-00.tikz}
    \quad\ZWtoZX\quad
    \input{figures/proofs/ZWAxioms/Z-copy.tikz} \\
    \quad\ZXfromXW\quad
    \input{figures/axioms/ZW/Z-copy-annot-01.tikz}
  \end{align*}
\end{proof}

\begin{lemma}
  \label{lem:Z-NF}
  The translation of the following ZW diagrams under $\iWX$ equal in $\ZXf$:
  \[
    \begin{tikzpicture}
	\begin{pgfonlayer}{nodelayer}
		\node [style=rZWdimension] (11)  at (0.0, 0.3) {$a~$};
		\node [style=white dot] (1)  at (0.0, 0.5) {};
		\node [style=none] (2)  at (0.0, -0.5) {};
	\end{pgfonlayer}
	\begin{pgfonlayer}{edgelayer}
		\draw (2.center) to (1);
	\end{pgfonlayer}
\end{tikzpicture}
    \qquad \qquad \qquad \qquad
    \input{figures/axioms/ZW/Z-NF-annot-01.tikz}
  \]
\end{lemma}
\begin{proof}
  For $\overrightarrow{C} = (\underbrace{0, \cdots, 0}_{a+1}, \frac{1}{1!a!}, \frac{1}{2!a!}, \cdots, \frac{1}{(a - 1)!a!})$,
  \begin{align*}
    \begin{tikzpicture}
	\begin{pgfonlayer}{nodelayer}
		\node [style=rZWdimension] (11)  at (0.0, 0.3) {$a~$};
		\node [style=white dot] (1)  at (0.0, 0.5) {};
		\node [style=none] (2)  at (0.0, -0.5) {};
	\end{pgfonlayer}
	\begin{pgfonlayer}{edgelayer}
		\draw (2.center) to (1);
	\end{pgfonlayer}
\end{tikzpicture}
    \quad\ZWtoZX\quad
    \input{figures/proofs/ZWAxioms/Z-NF.tikz} \\
    \quad\ZXfromXW\quad
    \input{figures/axioms/ZW/Z-NF-annot-02.tikz}
    \quad = \quad
    \input{figures/axioms/ZW/Z-NF-annot-01.tikz}
  \end{align*}
  Note that we can derive the last equality of the ZW diagrams using \cref{lem:ZW-a}.
\end{proof}

\rules*
\begin{proof}
  By the functoriality of $\iWX$, we only need to show that all the rewrite rules of $\ZWf$ are derivable in $\ZXf$.
  The rewrite rules are given in \cref{subsec:zw-axioms}, and the derivation of each axiom is given in~\crefrange{lem:Z-spider-rule}{lem:Z-NF}.
\end{proof}

\end{document}